\documentclass[18pt,english]{article}
\usepackage[utf8]{inputenc}
\usepackage[a4paper, total={6in, 9in}]{geometry}
\usepackage[pdftex]{graphicx}
\usepackage[table]{xcolor}
\usepackage{float}
\usepackage{hhline}
\usepackage{amssymb}
\usepackage{biblatex}
\usepackage{amsmath}
\usepackage{enumitem}
\usepackage{amsthm}
\usepackage{eurosym}
\usepackage{hyperref}
\usepackage{hhline}
\newtheorem{definition}{Definition}
\newtheorem{proposition}{Proposition}
\newtheorem{theorem}{Theorem}

\newtheorem{lemma}{Lemma}
\newtheorem{corollary}{Corollary}
\newcommand{\commentout}[1]{}
\theoremstyle{remark}
\newtheorem{remark}{Remark}

\DeclareMathOperator{\OPT}{OPT}

\bibliography{bibl}

\title{Sender-Receiver Games with Payment Commitments}
\author{Ivan Geffner \\ Utrecht University}
\date{}

\begin{document}

\maketitle

\begin{abstract}
We study sender-receiver games in which players commit, before the state is drawn, to payments contingent on the subsequently verified state and the receiver's action. Payments may be transferred or burned, serving as contractual guarantees or promised tips. We analyze equilibrium guarantees, welfare, and computational complexity.

For arbitrary finite action sets, we characterize pessimistic utility guarantees under unilateral commitments. Sender-only and sender-first commitments have worst-equilibrium welfare equal to the best constant-action welfare. Simultaneous commitments can instead reduce welfare below every equilibrium of the original game. Neither simultaneous nor sender-first commitments always admit an efficient equilibrium. Receiver-first commitments always admit an efficiently computable welfare-maximizing equilibrium, but can also reduce worst-equilibrium welfare. For binary actions, we characterize their entire equilibrium payoff region.

With binary actions, optimal receiver commitments are computable in linear time, whereas optimizing deterministic sender commitments is NP-hard. Privately randomized sender commitments admit an exact $O(n\log n)$ algorithm, with an optimal lottery over at most two schemes differing at only one state. This yields an additive approximation by deterministic commitments. For arbitrary finite action sets, optimal privately randomized sender payoffs equal Bayesian persuasion values: ordinary persuasion with burning alone, and persuasion with nonnegative compliance payments when transfers are allowed. These equivalences yield polynomial-time algorithms with polynomial-size commitment lotteries.
\end{abstract}

\section{Introduction}

Suppose a customer is considering buying a product
that may be defective. The shopkeeper knows its
condition but benefits from making a sale.
The customer may therefore distrust the recommendation
to buy: the shopkeeper might give the same advice
regardless of quality. As a result, a customer who
would buy a product known to be sound may decline
to purchase.

A guarantee can change these incentives. By promising
a refund or replacement if the product proves
defective, the shopkeeper protects the customer
and makes selling a defective product less attractive.
The shopkeeper remains free to recommend a purchase,
but the guarantee changes the financial consequences
of that advice.

Payments can also originate on the customer's side.
Consider an appliance whose warranty has expired.
A technician knows whether a repair would work
but earns more from selling a replacement.
The customer can promise a tip if the appliance
is repaired successfully, making it more attractive
for the technician to recommend a viable repair.
Here, the customer rewards a favorable outcome,
rather than receiving compensation for an
unfavorable one. Both examples use commitments
to outcome-contingent payments to change the
incentives underlying advice.

The intuition for why such payments should help
is straightforward. If changing an outcome benefits
one player more than it costs the other, the
beneficiary can offer compensation that makes
the change attractive to both. However, the
existence of a mutually beneficial payment scheme
does not ensure that players will choose it.
The commitments themselves must survive strategic
deviations.

Jackson and Wilkie~\cite{jackson2005endogenous}
demonstrate this difficulty in normal-form games
with binding side-payment commitments. Payments
can fail to support efficiency even when information
is complete and contracting is costless.
Indeed, an efficient equilibrium of the original
game can cease to be supportable once players
can commit to payments. The problem is that
players choose commitments to improve their own
payoffs. A player may profit by redirecting behavior
toward an outcome that gives it a larger payoff
but creates less total welfare. The ability to
compensate others does not eliminate the conflict
over which commitments to make.

Do guarantees and tips overcome this difficulty
in strategic communication, or can the contracting
stage undermine their benefits? We study this
question in sender-receiver games. A sender
privately observes a state and communicates with
a receiver, who chooses an action affecting both
players. Before the state is drawn, one or both
players can commit to nonnegative payments contingent
on the subsequently verified state and the
receiver's action. Payments can be transferred
to the other player or paid to an external sink,
which we call money burning. Guarantees and
promised tips are interpretations of these
enforceable commitments; subsequent communication
and action choices remain strategic.

Our analysis distinguishes two questions.
What can commitments guarantee under unfavorable
equilibrium selection, and what can they achieve
in a favorable equilibrium? We study both
unilateral commitments and bilateral commitments,
with the latter made simultaneously or sequentially
in either order. Our results show that the identity
of the committing player and the order of commitments
determine which benefits of contracting survive.

\paragraph{Welfare guarantees and commitment order.}

For arbitrary finite action sets, we characterize
the pessimistic utility guarantees under unilateral
commitments. The receiver can guarantee its original
babbling payoff, but cannot guarantee more.
The sender can secure the welfare of the best
constant action minus that receiver payoff.
Consequently, worst-equilibrium welfare under
sender-only commitments equals the welfare of
the best constant action.

This welfare guarantee survives when the receiver
can commit after observing the sender's scheme.
Sender-first commitments preserve the sender's
unilateral pessimistic guarantee and the receiver's
original babbling guarantee. Simultaneous commitments,
however, can introduce equilibria with lower welfare
than every equilibrium of the original game.
This deterioration can occur even when both
players' commitments are strict best responses.

Protecting a welfare baseline does not ensure
that efficient decisions are attainable.
We exhibit a binary-action game in which neither
simultaneous nor sender-first commitments improve
even the best equilibrium beyond original babbling
welfare, despite the availability of strictly
better state-dependent outcomes.

Receiver-first commitments yield a different
conclusion. For every finite action set, they
admit a welfare-maximizing equilibrium, computable
using $O(n|A|)$ arithmetic operations, where $n$
is the number of states. Yet their worst equilibria
can have lower welfare than those of the original
game. Thus, sender-first commitments protect
against unfavorable equilibrium selection without
always permitting efficiency, whereas receiver-first
commitments always permit efficiency without
necessarily protecting the original welfare baseline.

For binary actions, we further characterize
the entire receiver-first equilibrium payoff
region. The characterization identifies when
efficient welfare can be divided flexibly between
the players and when its division is uniquely
determined. It also gives exact worst-equilibrium
welfare and linear-time algorithms for implementing
feasible target payoffs and maximizing either
player's utility.

\paragraph{Computation and private randomization.}

Designing an optimal unilateral commitment reveals
a computational asymmetry. With binary actions,
an optimal receiver commitment and a supporting
equilibrium can be computed in linear time.
Optimizing deterministic sender commitments is
NP-hard, even with only two actions. The difficulty
arises because payments can change the sender's
preferred recommendation, creating discrete
changes in the receiver's incentives to follow
the advice.

Private randomization removes this computational
obstruction. The sender publicly commits to a
lottery over payment schemes and privately observes
its realization before communicating. For binary
actions, we give an exact $O(n\log n)$ algorithm.
An optimal lottery requires at most two schemes
that differ at only one state. This structure
also yields an additive approximation by
deterministic commitments.

For arbitrary finite action sets, private
randomization connects payment commitments
to Bayesian persuasion. Although our sender
cannot directly commit to disclosure, privately
realized contracts can provide incentives to
follow a chosen disclosure policy. We show that
optimal sender payoffs coincide with ordinary
persuasion when only burning is allowed, and
with persuasion with nonnegative compliance
payments when transfers are also available
\cite{dughmi2019persuasion}. The equivalence gives
polynomial-time algorithms and polynomial-size
commitment lotteries for arbitrary finite action sets.

Guarantees and tips can therefore do more than
compensate a player after an unfavorable decision:
they can shape the credibility of advice and
the equilibrium use of information. Their benefits
depend, however, on how commitment opportunities
are assigned. Our results distinguish the protocols
that protect welfare from those that make efficiency
attainable, and identify when optimal commitments
can be computed.

\subsection{Related Work}
\label{sec:related-work}

\paragraph{Endogenous payments and implementation.}
Our closest classical antecedent is Jackson and
Wilkie~\cite{jackson2005endogenous}, who study games
in which players make binding offers of
strategy-contingent transfers before choosing
their actions. They show that endogenous payments
need not produce efficiency, even with complete
information and costless contracting.
Subsequent work explores how the contracting
institution changes this conclusion.
Geffner, Oesterheld, and
Conitzer~\cite{geffner2025maximizing} introduce a
staged-commitment protocol that, under their
assumptions, implements welfare-maximizing payoffs
that strictly Pareto-improve a given equilibrium.
Liu, Farina, and Ozdaglar~\cite{liu2026internal}
study self-enforcing internal transfers rewarding
adherence to prescribed strategies, together
with a mediated extension using binding offers.
We examine endogenous payments in strategic
communication, where the informed sender and
the decision-making receiver have different roles.
This structure allows us to distinguish unilateral
guarantees, simultaneous commitments, and the two
sequential commitment orders.

Other payment mechanisms assign incentives to
an external designer. Monderer and
Tennenholtz~\cite{monderer2004implementation}
study $k$-implementation through nonnegative
subsidies, while Zhang et al.~\cite{zhang2024steering}
use external payments to steer no-regret learners.
Our payments are instead chosen and financed
by the players, so the commitments themselves
must be optimal. Money burning has also been
studied as part of the sender's signaling
behavior~\cite{austensmith2000cheap,kartik2007note}
and within mediated communication
\cite{liu2024burning}. In our model, burning is
committed to before the state is drawn and
depends on the subsequently verified state
and action.

\paragraph{Persuasion, advice, and computation.}
Our communication stage follows the cheap-talk
tradition of Crawford and Sobel
\cite{crawford1982strategic}, while Bayesian
persuasion allows the sender to commit directly
to disclosure~\cite{kamenica2011persuasion}.
Dughmi, Niazadeh, Psomas, and
Weinberg~\cite{dughmi2019persuasion} augment
persuasion with payments for following
recommendations. Their framework includes
nonnegative compliance payments and characterizes
the compensation required to support a given
recommendation rule.
Proposition~\ref{prop:persuasion-equivalence}
identifies their nonnegative-payment optimization
problem with our sender's optimistic problem
under private randomization.
Our contribution is the implementation through
outcome-contingent contracts and strategic
communication, including commitment-stage
incentives. For binary actions, we additionally
give an explicit $O(n\log n)$ algorithm and an
optimal lottery over at most two schemes differing
at only one state. Their two-action allocation
characterization with payments of either sign
and their nonnegative-payment results under
symmetric-action assumptions provide related
benchmarks.

Receiver commitments are closely related to
Krishna and Morgan~\cite{krishna2008contracting},
who study an uninformed principal that commits
to paying an informed adviser while retaining
decision-making authority. We allow either or
both parties to commit to verified-state-and-action
contingent payments and compare commitment orders.
On the computational side, Babichenko et al.
\cite{babichenko2024algorithmic} show that
sender-optimal cheap-talk equilibria can be
computed efficiently with binary receiver actions.
Our hardness result therefore comes from designing
the payment scheme, rather than optimizing
equilibrium selection in the original game.
The tractability gained through randomization
also parallels results for menus of contracts
by Castiglioni, Marchesi, and
Gatti~\cite{castiglioni2023menus}, although our
private contract realizations influence
communication rather than directly specifying
an agent's action incentives alone.

\paragraph{Commitments as mediating devices.}
More broadly, we study payment enforcement as a
mediating device: an institution that changes
the strategic capabilities available to players
and hence the outcomes they can sustain.
Other devices intervene through communication,
commitment, observability, or the available actions.
Communication mediators collect reports and issue
recommendations; work on implementing these devices
through cheap talk identifies when players can
reproduce their behavior without a trusted
intermediary
\cite{abraham2019implementing, abraham2008lower, geffner2024communication, geffner2023lower, gerardi2004unmediated}.
Comparisons of mediated communication, cheap talk,
and persuasion further distinguish the value
of recommendations from commitment to disclosure
\cite{arieli2023mediated,arieli2023long,arieli2023resilient,corrao2023bounds}.
Program equilibrium allows commitment to programs
that inspect one another's code
\cite{tennenholtz2004program,Oesterheld18:Robust}, while simulation
provides information about another player's
behavior~\cite{kovarik2023simulation, kovarik2026recursive}. Both allow players, for instance, to avoid the inefficiencies of the Prisoner's Dilemma.
Expanding action sets can instead strengthen
a given Nash equilibrium against self-enforcing
coalitional deviations while preserving the
original action-profile payoffs
\cite{geffner2024making}.
Our device enforces payments chosen by the players.
We ask how its benefits depend on who receives
access to this capability, the order in which
it is exercised, and equilibrium selection.

\section{Model}\label{sec:model}

A \emph{sender--receiver game} is a tuple
$\Gamma=(\Theta,A,p,u_S,u_R)$, where $\Theta$ is a finite set of
states, $A$ is a finite set of actions, $p\in\Delta(\Theta)$ is a
full-support prior, and $u_S,u_R:\Theta\times A\to\mathbb{R}$ are
the utility functions of the sender $S$ and receiver $R$.
We write $n:=|\Theta|$, $P:=\{S,R\}$, and $-i$ for the player
other than $i\in P$. An \emph{outcome} is a pair $(\theta,a)$.

Nature draws $\theta\sim p$ and reveals it only to the sender.
The sender sends a message $m$ from a finite message set $M$;
the receiver observes $m$ and chooses $a\in A$. Player $i$ then
receives $u_i(\theta,a)$ utility. Messages are costless and do not enter
either player's utility function. We fix a message set with
$|M|\geq\max\{n+1,|A|\}$, allowing distinct messages for states
or action recommendations. This is an assumption on the
communication technology, rather than a claim that arbitrary
equilibria can be compressed to a specified number of messages.

A behavioral strategy profile is
$\vec{\sigma}=(\sigma_S,\sigma_R)$, with
$\sigma_S:\Theta\to\Delta(M)$ and
$\sigma_R:M\to\Delta(A)$. Its induced outcome distribution is
\[
x_{\vec{\sigma}}(\theta,a)
:=p(\theta)\sum_{m\in M}
  \sigma_S(m\mid\theta)\sigma_R(a\mid m).
\]
We write
\[
u_i(\vec{\sigma})
:=\sum_{\theta,a}x_{\vec{\sigma}}(\theta,a)u_i(\theta,a),
\qquad
w(\theta,a):=u_S(\theta,a)+u_R(\theta,a),
\]
and $w(\vec{\sigma}):=u_S(\vec{\sigma})+u_R(\vec{\sigma})$.
An outcome is \emph{on path} if
$x_{\vec{\sigma}}(\theta,a)>0$ and \emph{off path} otherwise;
we denote the set of on-path outcomes by
$\operatorname{supp}(\vec{\sigma})$.

We call a profile \emph{babbling} if the sender's message
distribution is independent of the state. For each $a\in A$, let
\[
S_a(\Gamma):=\sum_\theta p(\theta)u_S(\theta,a),
\qquad
R_a(\Gamma):=\sum_\theta p(\theta)u_R(\theta,a),
\qquad
W_a(\Gamma):=S_a(\Gamma)+R_a(\Gamma).
\]
We sometimes omit the argument $\Gamma$ and write
$u_S(a)=S_a$, $u_R(a)=R_a$, or $W(a)=W_a$.
The \emph{babbling payoffs} are
\[
b_R(\Gamma):=\max_{a\in A}R_a(\Gamma),
\qquad
b_S(\Gamma):=
\max_{a\in\arg\max_{a'\in A}R_{a'}(\Gamma)}S_a(\Gamma).
\]
Thus the receiver chooses a prior-optimal action, breaking
payoff ties in the sender's favor. We call a selected action
attaining these two values the \emph{babbling default action}.

\subsection{Payment schemes}

A payment scheme for player $i$ specifies nonnegative payments
to the other player and to an external sink $\bot$ at each
outcome:
\[
\tau_i:\Theta\times A\longrightarrow
\mathbb{R}_{\geq0}^{\{-i,\bot\}}.
\]
Here $\tau_i^{-i}(\theta,a)$ is transferred to the other player,
and $\tau_i^\bot(\theta,a)$ is \emph{burned}. Write
$\|\tau_i(\theta,a)\|:=
\tau_i^{-i}(\theta,a)+\tau_i^\bot(\theta,a)$,
let $T_i$ denote the set of such schemes, and let $\mathbf{0}$
denote the zero scheme. Each promised payment is finite, but
we impose no common upper bound or budget constraint.

For $\tau=(\tau_S,\tau_R)$, the modified payoffs are
\[
u_i^\tau(\theta,a)
:=u_i(\theta,a)-\|\tau_i(\theta,a)\|
  +\tau_{-i}^{\,i}(\theta,a).
\]
We denote the resulting communication game by $\Gamma^\tau$.
All payoff benchmarks apply to this game using its modified
utilities; for example, $R_a(\Gamma^\tau)$ is the receiver's
expected net payoff from always choosing $a$.
We write $u_i^\tau(\vec{\sigma})$ and
$w^\tau(\vec{\sigma})$ for expected net utilities and welfare.
A \emph{transfer scheme} has $\tau_i^\bot\equiv0$; a
\emph{burning scheme} has $\tau_i^{-i}\equiv0$.

\begin{remark}\label{rem:welfare}
Transfers cancel when the players' utilities are added,
whereas burning reduces welfare:
\[
w^\tau(\vec{\sigma})
=w(\vec{\sigma})-
\sum_{\theta,a}x_{\vec{\sigma}}(\theta,a)
 \bigl(\tau_S^\bot(\theta,a)+\tau_R^\bot(\theta,a)\bigr).
\]
Throughout, welfare in a commitment extension means net
welfare, after any realized burning.
\end{remark}

\begin{remark}\label{rem:offpath}
There is no cost of announcing a payment scheme. Payments
are executed only at the realized outcome. A promise at an
off-path outcome therefore has no realized cost under the
specified profile, although it can affect incentives to
deviate. Whether a promise is off path depends on subsequent
play.
\end{remark}

\subsection{Commitment protocols}
\label{sec:extensions}

We consider five protocols, indexed by
\[
\mathcal{C}:=\{S,R,SR,S\to R,R\to S\}.
\]
The corresponding extension of $\Gamma$ is denoted by
$\Gamma^c$.
\begin{itemize}
\item In $\Gamma^S$ or $\Gamma^R$, only the indicated player
chooses a payment scheme; the other player's scheme is
$\mathbf{0}$.
\item In $\Gamma^{SR}$, both players choose their schemes
simultaneously, after which both schemes are made public.
\item In $\Gamma^{S\to R}$, the sender publicly commits first,
and the receiver observes this scheme before choosing and
publicly announcing its own scheme.
\item In $\Gamma^{R\to S}$, the receiver publicly commits first,
and the sender observes this scheme before choosing and
publicly announcing its own scheme.
\end{itemize}
We refer to the last two protocols as \emph{Stackelberg
commitments}, with the first player the leader and the second
the follower. The order concerns payment commitments; the
sender still sends the message and the receiver still chooses
the final action in both protocols.

All commitment choices are made before the state is drawn.
Once the commitment stage is complete, nature draws
$\theta\sim p$ and reveals it to the sender. The players then
play the communication game $\Gamma^\tau$. After the receiver
chooses an action, the state is verified for enforcement and
all promised payments at the realized outcome are executed.
Commitments are binding and cannot be revoked. In particular,
in a Stackelberg protocol the leader's scheme remains in
force after the follower's choice.

The timing prevents the choice of a scheme from signaling
private information about the state. Ex post verification is
a separate assumption: it makes state-contingent payments
enforceable without revealing the state to the receiver before
the action is chosen.
We also write $\Gamma^c_{\mathrm{tr}}$ and $\Gamma^c_\bot$ for
variants restricting all committing players to transfers or
burning, respectively.

\subsection{Privately randomized commitments}
\label{sec:random-commitments}

We additionally allow a player to commit publicly to a
finitely supported lottery over payment schemes. Write
$\Delta_f(T_i)$ for these lotteries and
$\Gamma^c_\Delta$ for the corresponding extension. The protocol
for choosing lotteries is the same as the protocol for choosing
schemes: choices are simultaneous under $SR$, and under
$i\to j$ the follower observes the leader's chosen lottery
before choosing its own. A noncommitting player uses the
point mass at $\mathbf{0}$.

After all lotteries have been chosen and announced, nature
draws each scheme $\tau_i\sim\rho_i$, independently across
players and independently of $\theta\sim p$. Each player
privately observes its own realized scheme, and the sender
also observes $\theta$. The sender then communicates and the
receiver chooses an action. The realized schemes and state
are verified for enforcement afterward.

Thus a follower observes the leader's \emph{distribution}
over schemes, not its realization. During communication,
strategies may depend on a player's private realization, and
beliefs must account for the other player's realization.
For a fixed public lottery profile $\rho$, denote this finite
continuation game by $\Gamma^\rho$.

This randomization is part of the enforceable contract, not
unobserved mixing over public commitment announcements.
Different realizations can induce different optimal
communication choices, so a lottery can randomize induced
behavior without requiring the sender to be indifferent
between messages under a single deterministic scheme.
We do not require these optimal choices to be strict.

\subsection{Equilibrium}
\label{sec:solution-concept}

Our solution concept requires optimality throughout the
extended game: at the commitment stage, after every observed
commitment, and during communication. We restrict attention
to pure strategies at public commitment decisions. In a
randomized-commitment extension, a pure commitment decision
chooses a lottery; behavioral randomization in the finite
continuation game is allowed in all extensions.

For a fixed public commitment history $h$, let $G(h)$ be the
continuation game: $\Gamma^\tau$ for deterministic schemes or
$\Gamma^\rho$ for lotteries. An assessment in $G(h)$ specifies
behavioral strategies and beliefs at information sets.
It is a \emph{sequential equilibrium} if strategies are
sequentially rational at every information set and beliefs
are limits of Bayes posteriors along completely mixed
strategy profiles converging to the specified strategies.
We write $\operatorname{SE}(G(h))$ for these assessments.

At communication and action decisions, we use favorable
tie-breaking: among choices maximizing its own expected net
payoff, a player chooses those maximizing the other player's
expected net payoff, evaluated at the same information set.
Randomization is permitted among choices tied under both
criteria. For binary actions, when both players are indifferent
in a given state, we use action $0$ in defining the
sender-preferred recommendation below. At commitment decisions,
optimality refers to the committing player's own expected
payoff; no additional tie-breaking restriction is imposed.

\begin{remark}\label{rem:consistency}
In a deterministic continuation $\Gamma^\tau$, the receiver's
beliefs at messages sent with positive probability are fixed
by Bayes' rule. At unused messages, any posterior over
$\Theta$ can be obtained as a limit of full-support trembles,
with smaller-order trembles for states assigned zero limiting
probability. The receiver must nevertheless respond optimally
to the specified posterior, and the sender must not benefit
from using the unused message. We impose no general requirement
that off-path beliefs equal the prior. In $\Gamma^\rho$,
consistency also concerns privately realized schemes and is
imposed through the finite-game definition above.
\end{remark}

\begin{remark}\label{rem:menu}
A message offers the sender the action lottery specified by
the receiver's response. Sender optimality compares these
lotteries across all available messages, including unused
ones. Receiver optimality is required at every message,
including those off path. In particular, splitting a mixed
receiver response into separate messages recommending its
individual actions need not preserve equilibrium.
\end{remark}

To define equilibrium at the commitment stage, select a
continuation assessment $e(h)\in\operatorname{SE}(G(h))$ after
\emph{every} possible completed public commitment history.
Let $F_i(h)$ be player $i$'s expected net payoff in this
assessment. For deterministic schemes write
$F_i(\tau_S,\tau_R)$, and for lotteries write
$F_i(\rho_S,\rho_R)$.

\begin{definition}\label{def:commitment-equilibrium}
A \emph{sequential equilibrium of the commitment extension}
consists of such a family of continuation assessments and
commitment strategies satisfying the following conditions.
\begin{enumerate}
\item For a unilateral protocol, the committing player's
chosen scheme maximizes its payoff $F_i$, with the other
scheme fixed at $\mathbf{0}$.
\item For simultaneous commitments, the chosen pair is a
Nash equilibrium of the induced commitment game:
\[
F_i(\tau_i^*,\tau_{-i}^*)
\geq F_i(\widehat\tau_i,\tau_{-i}^*)
\quad\text{for every }i\in P
\text{ and }\widehat\tau_i\in T_i,
\]
where arguments are placed in the appropriate player
coordinates.
\item For a Stackelberg protocol $i\to j$, the follower's
strategy specifies a response $r_j(\tau_i)\in T_j$ after
\emph{every} leader scheme. It satisfies
\[
r_j(\tau_i)\in
\arg\max_{\tau_j\in T_j}F_j(\tau_i,\tau_j)
\quad\text{for every }\tau_i\in T_i.
\]
The leader's chosen scheme satisfies
\[
\tau_i^*\in
\arg\max_{\tau_i\in T_i}
F_i\bigl(\tau_i,r_j(\tau_i)\bigr).
\]
Here $F_i$ and $F_j$ are written with the leader's coordinate
first for convenience.
\end{enumerate}
For $\Gamma^c_\Delta$, replace payment schemes by lotteries
in these conditions and use the assessments of $\Gamma^\rho$.
\end{definition}

This definition imposes sequential equilibrium in each finite
continuation and subgame-perfect optimality at the public
commitment stage. It specifies our use of sequential
equilibrium for an extension with a continuum of available
commitments, without invoking a finite-game tremble definition
at that stage. In particular, a deviation by a Stackelberg
leader is evaluated using the follower's specified optimal
response to that deviation. We do not assume that a follower
breaks commitment ties in the leader's favor.
We write $\operatorname{SE}(\Gamma^c)$ for the equilibria so
defined, and analogously for the randomized extensions.

\subsection{Optimal schemes and equilibrium values}

For a deterministic profile $\tau$, define the continuation
values
\[
V_i^+(\tau):=
\sup_{e\in\operatorname{SE}(\Gamma^\tau)}u_i^\tau(e),
\qquad
V_i^-(\tau):=
\inf_{e\in\operatorname{SE}(\Gamma^\tau)}u_i^\tau(e).
\]
These evaluate a fixed scheme profile under optimistic and
pessimistic continuation selection. For unilateral commitment,
let
\[
\OPT_S^\pm(\Gamma):=
\sup_{\tau_S\in T_S}V_S^\pm(\tau_S,\mathbf{0}),
\qquad
\OPT_R^\pm(\Gamma):=
\sup_{\tau_R\in T_R}V_R^\pm(\mathbf{0},\tau_R).
\]
Thus $\OPT_i^-$ is the best payoff that player $i$ can guarantee
across continuation equilibria by choosing its scheme, while
$\OPT_i^+$ is its best payoff under favorable continuation
selection. Define $\OPT_{i,\Delta}^\pm$ analogously by optimizing
over $\Delta_f(T_i)$ and using the continuation equilibria of
$\Gamma^\rho$.

For $c\in\mathcal C$ and an equilibrium $E$ of the entire
extension, write $u_i(E)$ for expected net utility and
$w(E):=u_S(E)+u_R(E)$. Whenever the equilibrium set is nonempty,
define, for $f\in\{u_S,u_R,w\}$,
\[
\mathrm{EQ}_{c,f}^-(\Gamma):=
\inf_{E\in\operatorname{SE}(\Gamma^c)}f(E),
\qquad
\mathrm{EQ}_{c,f}^+(\Gamma):=
\sup_{E\in\operatorname{SE}(\Gamma^c)}f(E).
\]
The notation $\mathrm{EQ}_{c,\Delta,f}^\pm$ refers to
$\Gamma^c_\Delta$. The extrema range over full equilibrium
strategies, including continuation selection and, for a
Stackelberg protocol, the follower's responses after all
leader commitments. We use suprema and infima without
assuming attainment unless explicitly stated.

\begin{lemma}\label{lem:opt-eq}
Fix $i\in P$. Suppose that both continuation extrema are
attained after every unilateral scheme and that both outer
optima defining $\OPT_i^\pm(\Gamma)$ are attained. Then
\[
\OPT_i^\pm(\Gamma)=\mathrm{EQ}_{i,u_i}^\pm(\Gamma).
\]
The same conclusion holds for privately randomized commitments
under the analogous attainment assumptions.
\end{lemma}

The distinction between these definitions matters:
$\OPT_i^\pm$ concerns unilateral scheme optimization, whereas
$\mathrm{EQ}_{c,f}^\pm$ evaluates equilibria in which the
commitments themselves must be optimal. We use the latter
notation for welfare and for all bilateral protocols.

\subsection{Binary actions and welfare benchmarks}

Many of our results concern binary actions, $A=\{0,1\}$.
Define
\[
\Delta_i(\theta):=
|u_i(\theta,0)-u_i(\theta,1)|.
\]
Let $a_S(\theta)$ maximize the sender's utility, breaking ties
in the receiver's favor and selecting $0$ if both players are
indifferent. Define $a_R(\theta)$ symmetrically. A state is
\emph{agreeing} if $a_S(\theta)=a_R(\theta)$ and
\emph{disagreeing} otherwise.

The \emph{sender-preferred recommendation profile} recommends
$a_S(\theta)$ and has the receiver follow this recommendation.
Its expected payoffs, whether or not it is an equilibrium, are
\[
d_i(\Gamma):=
\sum_\theta p(\theta)u_i(\theta,a_S(\theta)),
\qquad i\in P.
\]
If this profile can be completed to a sequential equilibrium,
we call it a \emph{sender-preferred obedient equilibrium}, or
simply an \emph{obedient equilibrium}. Recommending an action
need not reveal the state: different states may induce the
same recommendation. All these definitions also apply to
$\Gamma^\tau$ using its modified utilities.

\begin{proposition}\label{prop:binary-pareto}
Under favorable tie-breaking, a binary-action game satisfies:
\begin{itemize}
\item If $b_R(\Gamma)>d_R(\Gamma)$, a babbling equilibrium
weakly Pareto-dominates every sequential equilibrium.
\item If $b_R(\Gamma)\leq d_R(\Gamma)$, the sender-preferred
recommendation profile can be completed to a sequential
equilibrium that weakly Pareto-dominates every sequential
equilibrium.
\end{itemize}
\end{proposition}

Indeed, the receiver's two obedience inequalities for the
recommendation profile are precisely
$d_R(\Gamma)\geq R_0(\Gamma)$ and
$d_R(\Gamma)\geq R_1(\Gamma)$. This characterization concerns
the payoffs attainable in equilibrium; it does not assert
that babbling and obedient profiles are the only equilibria.
It applies to deterministic continuation games. A privately
randomized commitment introduces additional private information
and must be analyzed in its enlarged continuation game.

Finally, define
\[
W_{\mathrm{const}}(\Gamma):=\max_{a\in A}W_a(\Gamma),
\qquad
W^*(\Gamma):=\sum_\theta p(\theta)\max_{a\in A}w(\theta,a).
\]
The first is the best welfare from a constant action, and the
second is the maximum welfare attainable before burning.
We call a commitment equilibrium \emph{efficient} if its net
welfare equals $W^*(\Gamma)$.
An action rule $f:\Theta\to A$ is \emph{efficient} if
$w(\theta,f(\theta))=\max_{a\in A}w(\theta,a)$ in every
state, and \emph{nonconstant} if it assigns different actions to at least two states.
For binary actions, we say that $a$ \emph{weakly socially
dominates} $1-a$ if
$w(\theta,a)\geq w(\theta,1-a)$ in every state, and that it
\emph{strictly socially dominates} $1-a$ if all these
inequalities are strict. An action can weakly socially dominate
the other while nonconstant efficient action rules still exist
because of ties. Whenever this distinction matters, we state
explicitly whether a nonconstant efficient rule exists.

\subsection{Computational conventions}

For computational results, prior probabilities and payoffs
are rational numbers given explicitly in binary.
The stated running-time bounds count arithmetic operations
and comparisons. Our algorithms also have polynomial
bit complexity, as justified in their respective proofs.

An algorithm that computes an equilibrium outputs a finite
description of the equilibrium-path commitments and
continuation strategies, together with rules specifying
play after every off-path commitment history. In a
Stackelberg protocol, these rules also specify the
follower's commitment after every leader commitment.
Thus, the output describes contingent behavior without
explicitly listing the infinitely many possible
commitment histories.

\section{Main Results}

We first study guarantees under unfavorable equilibrium
selection, then compare bilateral commitment protocols,
and finally consider optimistic optimization.
Unless explicitly restricted to binary actions, the
results hold for arbitrary finite action sets.
Table~\ref{tab:main-results} summarizes the conclusions.
Throughout, $n$ denotes the number of states, and
running-time bounds count arithmetic operations.
We explain the main constructions and proof ideas below;
full proofs appear in the Appendix.

\begin{table}[H]
\centering
\small
\setlength{\tabcolsep}{4pt}
\renewcommand{\arraystretch}{1.25}

\textbf{(a) Unilateral commitments}
\medskip

\begin{tabular}{@{}p{0.22\linewidth}
                  p{0.35\linewidth}
                  p{0.35\linewidth}@{}}
\hline
Committing player
& Pessimistic utility
& Optimistic optimization\\
\hline
Receiver
& $b_R(\Gamma)$; attained in $O(n|A|)$.
& $O(n)$ with binary actions\\
Sender
& $W_{\mathrm{const}}(\Gamma)-b_R(\Gamma)$;
  \par attained in $O(n|A|)$
& NP-hard with deterministic schemes, even for binary
actions; polynomial time with private randomization,
and $O(n\log n)$ for binary actions\\
\hline
\end{tabular}

\bigskip
\textbf{(b) Bilateral commitments}
\medskip

\begin{tabular}{@{}p{0.22\linewidth}
                  p{0.35\linewidth}
                  p{0.35\linewidth}@{}}
\hline
Commitment order
& Worst-equilibrium welfare
& Efficient equilibrium\\
\hline
Simultaneous
& Can fall below original babbling \par welfare
& Need not exist\\
Sender first
& $W_{\mathrm{const}}(\Gamma)$
& Need not exist\\
Receiver first
& Can fall below original babbling \par welfare
& Always exists; computable in $O(n|A|)$.\\
\hline
\end{tabular}

\caption{Summary of results for games with
$n$ states. In panel (a), utility and optimistic optimization
refer to the committing player. Sender-only commitments
also satisfy
$\mathrm{EQ}_{S,w}^{-}(\Gamma)=W_{\mathrm{const}}(\Gamma)$.
In panel (b), commitments are deterministic.
For binary actions, receiver-first commitments
additionally admit a complete payoff-region
characterization and $O(n)$ algorithms for
implementing a feasible target profile or
maximizing either player's utility.}
\label{tab:main-results}
\end{table}

\subsection{Pessimistic Guarantees}

A unilateral commitment must protect the committing
player against every continuation equilibrium.
For the receiver, the original babbling payoff is
both a lower and an upper bound on this guarantee.

\begin{proposition}\label{prop:rec-pessimistic}
For every sender-receiver game $\Gamma$,
\[
    \OPT_R^{-}(\Gamma)=b_R(\Gamma).
\]
The zero-payment scheme attains this guarantee.
An equilibrium of $\Gamma^R$ attaining this payoff
is obtained by choosing this scheme and prescribing
a babbling continuation after every receiver commitment.
The equilibrium-path action can be computed in
$O(n|A|)$ arithmetic operations.
\end{proposition}

The receiver can always decline to make payments,
ignore the message, and choose an original prior-best
action. Conversely, every receiver commitment admits
a babbling continuation. In such a continuation,
the receiver's own transfers and burns can only
reduce its payoff from the selected constant action.
Commitments may therefore improve the receiver's
best equilibrium payoff, but cannot improve its
pessimistic guarantee. The full proof appears in
Appendix~\ref{sec:proof-rec-pessimistic}.

The sender faces a different problem. Her payments
can change which constant action the receiver prefers,
allowing her to compensate the receiver for choosing
an action with higher total welfare.

\begin{proposition}\label{prop:sen-pessimistic}
For every finite sender-receiver game $\Gamma$,
\[
    \OPT_S^{-}(\Gamma)
    =
    W_{\mathrm{const}}(\Gamma)-b_R(\Gamma).
\]
An optimal scheme can be computed in $O(n|A|)$
arithmetic operations. It transfers
\[
    t:=b_R(\Gamma)-R_c(\Gamma)
\]
on an action $c$ maximizing $W_c(\Gamma)$ and uses
only off-path burning on the other actions.
Every continuation equilibrium under this scheme
plays $c$ with probability one.
\end{proposition}

The transfer makes $c$ optimal for the receiver
under the prior, leaving the sender the remaining
constant-action welfare. Prior optimality alone,
however, does not rule out unfavorable informative
continuations.

To exclude them, the sender commits to burns that
give her the same strict ranking of actions in
every state, with $c$ ranked first. Favorable receiver
tie-breaking then makes every receiver response
pure. Since all sender types rank these responses
identically, all on-path messages induce the same
action. Receiver optimality and favorable
tie-breaking force that action to be $c$.
The burns are therefore never realized.
For the matching upper bound, any sender scheme
admits a babbling continuation, whose total welfare
is at most $W_{\mathrm{const}}(\Gamma)$ and whose
receiver payoff is at least $b_R(\Gamma)$. The full proof appears in
Appendix~\ref{sec:proof-sen-pessimistic}.

These individual guarantees yield an exact
worst-equilibrium welfare bound.

\begin{corollary}\label{cor:sw-pessimistic}
For every finite sender-receiver game $\Gamma$,
\[
    \mathrm{EQ}_{S,w}^{-}(\Gamma)
    =
    W_{\mathrm{const}}(\Gamma).
\]
An equilibrium attaining this welfare can be computed
in $O(n|A|)$ arithmetic operations.
\end{corollary}

In every equilibrium of the sender-only extension,
the sender can secure her commitment guarantee,
while the receiver can still secure its original
babbling payoff. Adding these guarantees gives
the lower bound. Equality is supported by the
scheme above and babbling continuations after
other commitments. Thus, even unfavorable equilibrium
selection preserves the welfare of the best
constant action, which weakly exceeds original
babbling welfare.

\subsection{Bilateral Commitments: Guarantees and Obstructions}

Allowing the receiver to commit as well changes
the incentives at the commitment stage. Under
simultaneous commitments, the receiver's scheme
remains fixed when the sender deviates. It can
therefore make the sender's alternative commitments
less attractive and support an outcome with
lower welfare.

\begin{proposition}\label{prop:simultaneous-welfare}
There exists a sender-receiver game $\Gamma$ with two states
and binary actions such that every equilibrium of the
original game has welfare $6$, while
\[
    \mathrm{EQ}_{SR,w}^{-}(\Gamma)\leq5
    <6
    =
    \mathrm{EQ}_{SR,w}^{+}(\Gamma)
    =
    W^*(\Gamma).
\]
The equilibrium witnessing the welfare decrease has strict
commitment-stage best responses for both players.
\end{proposition}

The example has equally likely states and the following
payoffs, listed as $(u_S,u_R)$:
\[
\begin{array}{c|cc}
    &0&1\\ \hline
    L&(4,4)&(0,0)\\
    R&(4,0)&(0,2)
\end{array}
\]
Without commitments, every equilibrium selects action
$0$ and has welfare $6$. Under simultaneous commitments,
the sender burns $5$ and the receiver burns $12$
at $(R,0)$. The resulting preferences support
recommendations $0$ in state $L$ and $1$ in state $R$,
giving payoffs $(2,3)$ and welfare $5$.
No burning occurs on path: the welfare loss comes
from choosing the less efficient action in state $R$.

Babbling continuations after commitment deviations
support this outcome. The receiver's burn makes it
too costly for the sender to induce constant action
$0$, while the sender's burn makes the receiver's
babbling alternatives unattractive. The full proof
checks arbitrary unilateral payment deviations and
shows that both prescribed commitments are strict
best responses and can be found in Appendix~\ref{sec:proof-worst-welfare}.

Sequencing commitments with the sender first restores
the sender-only guarantee, despite allowing the
receiver to respond with its own scheme.

\begin{theorem}\label{thm:sender-first-guarantees}
Let $\Gamma$ be a sender-receiver game with a finite action
set. Every sequential equilibrium of $\Gamma^{S\to R}$
gives the sender at least
$W_{\mathrm{const}}(\Gamma)-b_R(\Gamma)$ expected utility
and the receiver at least $b_R(\Gamma)$ expected utility.
Moreover,
\[
    \mathrm{EQ}_{S\to R,w}^{-}(\Gamma)
    =W_{\mathrm{const}}(\Gamma).
\]
An equilibrium attaining this welfare can be computed
using $O(n|A|)$ arithmetic operations.
\end{theorem}

The sender again uses off-path burns to create
a common strict ranking of actions across states.
The new difficulty is that receiver transfers can
change that ranking. The construction therefore
places sufficiently large gaps between successive
sender payoffs and slightly overcompensates the
receiver for selecting the intended constant action.

The key argument relates the receiver's payments
to the extent of nonconstant play. Sender message
incentives and favorable receiver tie-breaking imply
that substantial departures from a constant action
require substantial receiver transfers. When the
ranking gaps are sufficiently large, the cost of
these transfers exceeds any possible receiver gain
from changing the outcome. Every optimal receiver
response consequently preserves the intended action.
Taking the overcompensation arbitrarily close to zero
gives the stated sender guarantee. A separate
babbling construction attains the welfare bound
exactly. The full proof can be found in Appendix~\ref{sec:proof-sender-first-guarantees}.

This guarantee protects the best constant-action
welfare, but does not ensure that information is
used efficiently. The same binary-action example 
can obstruct efficiency under both simultaneous
and sender-first commitments.

\begin{proposition}\label{prop:bilateral-efficiency-obstruction}
There exists a sender-receiver game $\Gamma$ with two
states and binary actions such that
\[
    \mathrm{EQ}_{SR,w}^{+}(\Gamma)
    =
    \mathrm{EQ}_{S\to R,w}^{+}(\Gamma)
    =
    b_S(\Gamma)+b_R(\Gamma)
    =
    \frac52
    <
    3
    =
    W^*(\Gamma).
\]
\end{proposition}

In the example, the players disagree in both states.
The sender obtains $3$ from her preferred action
and $0$ from the other action; the receiver obtains
$2$ from its preferred action and $0$ from the other.
Their preferred actions reverse across the two
equally likely states. Efficiency therefore requires
always choosing the sender's preferred action,
whereas either constant action has welfare $5/2$.

The obstruction comes from commitment-stage incentives.
The sender can secure a positive payoff without
making payments, while the receiver can decline
its own commitments and benefit from any incoming
sender transfers. Combining these deviation
guarantees bounds the probability of selecting
the sender's preferred action and rules out any
welfare improvement over $5/2$.
The argument applies to every equilibrium under
either protocol, rather than only to a particular
choice of continuation equilibria. This is fully analyzed in Appendix~\ref{sec:prop-bilateral-efficiency-obstruction}.

\subsection{Receiver-First Commitments}

Reversing the commitment order changes the efficiency
conclusion. The receiver can use its initial scheme
to shape the sender's subsequent commitment incentives.

\begin{theorem}\label{thm:opt-welfare}
For every finite sender-receiver game $\Gamma$,
\[
    \mathrm{EQ}_{R\to S,w}^{+}(\Gamma)=W^*(\Gamma).
\]
An efficient equilibrium can be computed in
$O(n|A|)$ arithmetic operations.
\end{theorem}

The construction distinguishes constant and nonconstant
efficient action rules. If an efficient constant action
exists, the receiver makes no payments and the sender
compensates it for selecting that action. Suitable
babbling continuations support this outcome.

For a nonconstant rule, every constant action is
avoided in at least one state. The receiver can
therefore place large penalties on selected outcomes
without paying them along the intended rule.
These penalties make the sender's babbling alternatives
unattractive. On-path transfers divide the welfare,
while additional off-path burns make the intended
recommendations credible.

Commitment deviations are controlled in both directions.
After a different sender commitment, play switches
to a babbling continuation whose payoff does not
exceed the sender's prescribed payoff. After a
different receiver commitment, the sender optimally
responds within a prescribed family of babbling
continuations, leaving the receiver at most its
original babbling payoff. These responses are
optimal in the corresponding continuation games,
so the construction supports an equilibrium of
the entire extension. For a full proof, see Appendix~\ref{sec:receiver-first-constructions} and \ref{sec:thm-opt-welfare}.

For binary actions, we can characterize not only
efficient outcomes but every implementable payoff
profile. Define
\[
    W_{\mathrm{nc}}(\Gamma)
    :=
    \max_{\substack{f:\Theta\to\{0,1\}\\
                    f\text{ nonconstant}}}
    \mathbb{E}_{\theta\sim p}[w(\theta,f(\theta))],
\]
with $W_{\mathrm{nc}}(\Gamma)=-\infty$ if no
nonconstant rule exists. This is the largest
original welfare attainable by a nonconstant
deterministic action rule. In particular, such
a rule is efficient exactly when it attains $W^*(\Gamma)$.

Also define
\[
    \underline S:=\min_{a\in\{0,1\}}S_a,
    \qquad
    \overline R:=\max_{a\in\{0,1\}}R_a=b_R(\Gamma),
    \qquad
    \underline S_R
    :=
    \min_{\substack{a\in\{0,1\}\\R_a=\overline R}}S_a.
\]
Here, $\underline S$ is the sender's worst original
constant-action payoff, and $\underline S_R$ is
her worst payoff among the receiver's original
prior-best actions.

\begin{theorem}\label{thm:receiver-first-payoffs}
Let $\Gamma$ be a sender-receiver game with binary actions.
An expected utility profile $(U_S,U_R)$ is achieved by
a sequential equilibrium of $\Gamma^{R\to S}$ if and
only if it belongs to
\[
    \mathcal U_{\mathrm{nc}}(\Gamma)
    \cup
    \mathcal U_{\mathrm c}(\Gamma),
\]
where
\[
\begin{aligned}
    \mathcal U_{\mathrm{nc}}(\Gamma)
    &:=
    \left\{
        (U_S,U_R):
        U_S\geq\underline S,\;
        U_R\geq\overline R,\;
        U_S+U_R\leq W_{\mathrm{nc}}(\Gamma)
    \right\},\\
    \mathcal U_{\mathrm c}(\Gamma)
    &:=
    \left\{
        (U_S,\overline R):
        \underline S_R\leq U_S
        \leq W_{\mathrm{const}}(\Gamma)-\overline R
    \right\}.
\end{aligned}
\]
\end{theorem}

The nonconstant region combines two participation
constraints with a welfare bound. The receiver
can secure $\overline R$ by making no payments.
The sender can secure at least $\underline S$
by making no payments and sending a fixed message.
The nonconstant construction described above
implements every profile satisfying these bounds,
using on-path transfers to divide welfare and
burning any surplus not assigned to either player.

The constant region is more restrictive. After
the receiver's commitment, the sender can compensate
the receiver for choosing a constant action.
This deviation opportunity, together with the
receiver's participation constraint, forces the
receiver's payoff at a constant outcome to equal
$\overline R$. The remaining bounds determine
the sender's possible share.

For necessity, the binary-action structure allows
any equilibrium continuation to be replaced by a
deterministic action rule preserving both players'
expected payoffs. The proof (found in Appendix~\ref{sec:receiver-first-constructions} and \ref{sec:thm-receiver-first-payoffs}) then separates constant
and nonconstant rules and applies the corresponding
bounds. This is where the restriction to binary
actions enters the full characterization.

The payoff region immediately determines the
worst-equilibrium welfare.

\begin{corollary}\label{cor:receiver-first-worst-welfare}
For every sender-receiver game $\Gamma$ with binary actions,
\[
    \mathrm{EQ}_{R\to S,w}^{-}(\Gamma)
    =
    \begin{cases}
        \underline S+\overline R,
        &\text{if }\underline S+\overline R
                    \leq W_{\mathrm{nc}}(\Gamma),\\[1mm]
        \underline S_R+\overline R,
        &\text{otherwise}.
    \end{cases}
\]
\end{corollary}

When the nonconstant region is nonempty, it contains
the pair of participation payoffs
$(\underline S,\overline R)$, which minimizes welfare.
Otherwise, only the constant region remains, whose
lowest sender payoff is $\underline S_R$.

Two further consequences describe efficient
payoff divisions.

\begin{corollary}\label{cor:receiver-first-nonconstant}
Suppose $\Gamma$ has binary actions and admits an efficient
nonconstant action rule. Then $(U_S,U_R)$ is achieved by
a sequential equilibrium of $\Gamma^{R\to S}$ if and
only if
\[
    U_S\geq\underline S,\qquad
    U_R\geq b_R(\Gamma),\qquad
    U_S+U_R\leq W^*(\Gamma).
\]
\end{corollary}

Here, $W_{\mathrm{nc}}=W^*$, so participation and
total welfare are the only restrictions.
If no nonconstant efficient rule exists, the
division of efficient welfare is instead unique.

\begin{corollary}\label{cor:receiver-first-constant}
Suppose $\Gamma$ has binary actions and admits no efficient
nonconstant action rule. Then the unique expected utility
profile achieved by welfare-maximizing sequential equilibria
of $\Gamma^{R\to S}$ is
\[
    \bigl(W^*(\Gamma)-b_R(\Gamma),\ b_R(\Gamma)\bigr).
\]
\end{corollary}

Although receiver-first commitments always permit
efficiency, they need not preserve welfare in every
equilibrium.

\begin{corollary}\label{cor:receiver-first-welfare-loss}
There exists a sender-receiver game $\Gamma$ with two
states and binary actions such that every equilibrium
of the original game has welfare $6$, whereas
\[
    \mathrm{EQ}_{R\to S,w}^{-}(\Gamma)=2
    <6
    =\mathrm{EQ}_{R\to S,w}^{+}(\Gamma).
\]
\end{corollary}

The game in Proposition~\ref{prop:simultaneous-welfare}
also witnesses this loss. It has
$\underline S=0$, $\overline R=2$, and a nonconstant
rule with welfare $5$. The characterization therefore
implements the participation payoffs $(0,2)$,
while Theorem~\ref{thm:opt-welfare} ensures that
welfare $6$ remains attainable.

Finally, the characterization is constructive.

\begin{corollary}\label{cor:receiver-first-computation}
There exists an $O(n)$ algorithm that, given a
sender-receiver game $\Gamma$ with binary actions and
$n$ states and a target utility profile $(U_S,U_R)$,
decides whether this profile is implementable by a
sequential equilibrium of $\Gamma^{R\to S}$ and,
if so, constructs such an equilibrium.

Moreover, a sequential equilibrium maximizing social
welfare, the sender's expected utility, or the receiver's
expected utility can be computed in $O(n)$
arithmetic operations.
\end{corollary}

The constant-action benchmarks require a single
pass through the states. To compute $W_{\mathrm{nc}}$,
first choose a welfare-maximizing action in each state.
If this rule is constant and there are at least two
states, switch the action in a state where doing so
causes the smallest expected welfare loss.
The resulting value and the constant-action
benchmarks describe two simple payoff regions.
Membership and optimization over their union then
require only a constant number of additional
comparisons, followed by the corresponding
implementation construction.

\subsection{Optimistic Optimization}

We now turn to the largest equilibrium payoff
available to a unilateral committing player.
For binary actions, Proposition~\ref{prop:binary-pareto}
reduces the deterministic problem to comparing
babbling with sender-preferred obedient continuations.
Despite this common structure, sender and receiver
optimization have different computational properties.

\begin{theorem}\label{thm:opt-rec}
There exists an $O(n)$ algorithm that, given a sender-receiver
game $\Gamma$ with binary actions and $n$ states, computes
a sequential equilibrium of $\Gamma^R$ that gives the
receiver $\OPT_R^{+}(\Gamma)$ expected utility.
\end{theorem}

Receiver payments cannot improve its own payoff
in a babbling continuation. In an obedient
continuation, the receiver can instead compensate
the sender for recommending a different action.
At a disagreeing state, changing the sender's
preferred action costs $\Delta_S(\theta)$ and
benefits the receiver by $\Delta_R(\theta)$
before payment. The algorithm makes this change
when $\Delta_R(\theta)>\Delta_S(\theta)$ and
compares the resulting receiver payoff with
$b_R(\Gamma)$. If it exceeds that benchmark,
the recommendation profile is obedient; otherwise,
zero payments and babbling are optimal.
All calculations are statewise and take linear time. See Appendix~\ref{sec:proof-thm-opt-rec} for a full proof.

Sender optimization involves a discrete selection
problem.

\begin{theorem}\label{thm:opt-send-hard}
Given a sender-receiver game $\Gamma$ with binary actions
and rationally encoded data, computing
$\OPT_S^{+}(\Gamma)$ is NP-hard.
In particular, deciding whether
$\OPT_S^{+}(\Gamma)\geq v$ for a given rational threshold
$v$ is NP-hard.
\end{theorem}

The reduction encodes a knapsack item in each
of a collection of disagreeing states. By promising
a payment on her originally preferred action,
the sender can make herself willing to recommend
the receiver's preferred action instead.
This sacrifices a prescribed amount of sender
utility and creates a prescribed gain for the
receiver. The sender must select enough such
changes to make recommendations obedient while
limiting her own loss.

Payments below the switching threshold can also
affect obedience, but only continuously.
The reduction scales the receiver's utility gains
so that these smaller adjustments cannot bridge
the gap left by an insufficient collection of
switched states. Additional states control the
babbling alternatives. Thus an equilibrium meeting
the sender's target payoff exists exactly when
the corresponding knapsack instance is feasible. A full proof can be found in Appendix~\ref{sec:proof-opt-send-hard}.

Private randomization makes these changes divisible.
The sender publicly commits to a lottery over
payment schemes and privately observes its
realization before communicating.

\begin{theorem}\label{thm:opt-send-rand}
There exists an $O(n\log n)$ algorithm that, given a
sender-receiver game $\Gamma$ with binary actions and
$n$ states, computes a sequential equilibrium of
$\Gamma^S_\Delta$ giving the sender
$\OPT_{S,\Delta}^{+}(\Gamma)$ expected utility.

The equilibrium commitment is supported on at most
two payment schemes, which coincide at every state
except possibly one. In the constructed equilibrium,
the receiver's action is deterministic conditional
on each state except possibly that same state.
\end{theorem}

There are two ways to make recommendations acceptable:
change the recommended actions, or compensate the
receiver for following them. Changing the recommendation
at a disagreeing state sacrifices
$p(\theta)\Delta_S(\theta)$ of original sender utility
and gains $p(\theta)\Delta_R(\theta)$ of original
receiver utility. Private randomization permits
any fraction of this change.

The transfer cost also has a simple structure.
If a proposed action allocation gives the receiver
original expected utility $v$, the minimum total
expected transfer needed for obedience is
\[
    (R_0-v)_+ +(R_1-v)_+,
    \qquad
    x_+:=\max\{x,0\}.
\]
Each term compensates for one of the receiver's
constant-action alternatives. The proof shows
both that these payments are necessary and that
suitable privately realized contracts attain them,
using off-path burns to enforce the prescribed
recommendations.

The algorithm therefore sorts disagreeing states
by $\Delta_R(\theta)/\Delta_S(\theta)$ and processes
recommendation changes in that order, balancing
their sender cost against the reduction in required
transfers. Since the transfer cost is piecewise
linear in $v$, at most one state needs to be
partially switched. Sorting accounts for the
$O(n\log n)$ running time. The proof bounds
arbitrary randomized commitments, so the
optimization also covers babbling and commitments
that induce a different constant action (see Appendix~\ref{sec:proof-thm-send-rand}).

Restricting randomization to one state also
bounds its advantage over deterministic schemes.

\begin{corollary}\label{cor:randomization-gap}
For every sender-receiver game $\Gamma$ with binary actions,
\[
    0\leq
    \OPT_{S,\Delta}^{+}(\Gamma)-\OPT_S^{+}(\Gamma)
    \leq
    \max_{\theta\in\Theta}p(\theta)\Delta_S(\theta).
\]
Moreover, an equilibrium of $\Gamma^S$ whose sender
payoff is at least
\[
    \OPT_{S,\Delta}^{+}(\Gamma)
    -\max_{\theta\in\Theta}p(\theta)\Delta_S(\theta)
\]
can be computed in $O(n\log n)$ arithmetic operations.
\end{corollary}

Round the partially switched state toward the
receiver's preferred original action. This reduces
the sender's original payoff by at most that
state's expected utility difference, while weakly
reducing the compensation required for obedience.
Comparing the resulting construction with the
sender's pessimistic guarantee gives a deterministic
commitment supported by an equilibrium of the
entire extension. The formal argument is given
in Appendix~\ref{sec:cor-random}.

For arbitrary finite action sets, private
randomization also permits polynomial-time
optimization. Section~\ref{sec:persuasion-connection}
explains this through an equivalence with Bayesian
persuasion and gives a polynomial-size implementation.
The binary-action theorem above provides the
sharper running time and two-scheme support bound.

\section{Connection to Bayesian Persuasion}
\label{sec:persuasion-connection}

Our model gives players the ability to commit to payments,
while leaving communication strategic. Bayesian persuasion
instead gives the sender the ability to commit directly
to how information is disclosed. Despite this difference,
privately randomized sender commitments can attain the
same optimal payoffs as appropriate persuasion benchmarks.
This section explains that connection.

\paragraph{Commitment to disclosure.}
In Bayesian persuasion~\cite{kamenica2011persuasion},
the sender and receiver have the same information and
payoff structure as in a sender-receiver game:
the sender observes the state, and the receiver chooses
an action. The difference is that, before the state
is drawn, the sender publicly commits to a
\emph{signaling scheme}. This scheme specifies a
distribution over messages for each possible state.
Once the state is realized, the message is generated
according to the committed scheme; the sender cannot
replace it with a more favorable message.

The receiver observes the message, updates its beliefs,
and chooses an optimal action. The sender therefore
chooses how much information to disclose, taking the
receiver's subsequent incentives into account.
A signaling scheme can reveal the state, disclose
nothing, or provide partial information.

For the sender's optimization problem, it suffices
to consider direct schemes whose messages recommend
actions and whose recommendations the receiver is
willing to follow. Write $\pi(a\mid\theta)$ for the
probability of recommending action $a$ in state $\theta$.
The receiver's obedience conditions are
\[
    \sum_{\theta}p(\theta)\pi(a\mid\theta)
    \bigl(u_R(\theta,a)-u_R(\theta,b)\bigr)
    \geq0
    \qquad\text{for all }a,b\in A.
\]
Each inequality says that, conditional on recommendation
$a$, following it is at least as good as choosing $b$.
The expression is left unnormalized so that it also
applies to recommendations occurring with probability zero.

Let $\mathrm{BP}(\Gamma)$ denote the optimal sender
payoff over such schemes:
\[
    \mathrm{BP}(\Gamma)
    :=
    \max_{\pi\text{ obedient}}
    \sum_{\theta,a}
    p(\theta)\pi(a\mid\theta)u_S(\theta,a).
\]
As usual, the optimization uses favorable receiver
tie-breaking.

\paragraph{Persuasion with payments.}
Dughmi, Niazadeh, Psomas, and Weinberg
\cite{dughmi2019persuasion} augment this disclosure
commitment with payments for following recommendations.
In their nonnegative-payment setting, the sender also
promises an amount $z_a\geq0$ whenever action $a$ is
recommended and the receiver follows it.
If the receiver chooses another action, no payment
is made.

These payments give the sender an additional way
to influence behavior: a recommendation can be made
acceptable by compensating the receiver for following it.
For each $b\neq a$, obedience becomes
\[
    \sum_{\theta}p(\theta)\pi(a\mid\theta)
    \bigl(u_R(\theta,a)+z_a-u_R(\theta,b)\bigr)
    \geq0.
\]
The sender pays for this additional incentive, obtaining
expected utility
\[
    \sum_{\theta,a}p(\theta)\pi(a\mid\theta)
    \bigl(u_S(\theta,a)-z_a\bigr).
\]
We denote the maximum of this expression over obedient
schemes and nonnegative payments by
$\mathrm{BP}_{+}(\Gamma)$.

\paragraph{Implementing disclosure through payment commitments.}
In our model, the sender cannot directly commit to a
signaling scheme. However, a privately realized payment
scheme can determine which recommendation she prefers
in each state.

To see the idea, consider a deterministic action rule
$f:\Theta\to A$. The sender can commit to sufficiently
large burns on every outcome $(\theta,b)$ with
$b\neq f(\theta)$, while burning nothing on
$(\theta,f(\theta))$. If the receiver follows the
recommendations, the sender strictly prefers to
recommend $f(\theta)$ in every state. These burns
are never realized.

A randomized recommendation rule can similarly be
represented as a lottery over deterministic action
rules. The sender commits to the corresponding lottery
over burning schemes and privately observes the
selected scheme before communicating. The receiver
observes the recommendation but not the selected rule.
Its information about the state therefore matches
that generated by the corresponding persuasion scheme.

Compliance payments can be implemented in the same way.
Under the scheme associated with $f$, the sender
transfers $z_{f(\theta)}$ on the prescribed action
and nothing on other actions. Additional off-path
burning makes the prescribed action strictly preferable
for the sender even after accounting for this transfer.
Although payments in our model do not depend directly
on messages, the realized contract thus reproduces
the payments associated with following a recommendation.

Write $\OPT_{S,\Delta,\bot}^{+}(\Gamma)$ for the sender's
optimistic value when privately randomized commitments
are allowed but payments are restricted to burning.

\begin{proposition}\label{prop:persuasion-equivalence}
For every finite sender-receiver game $\Gamma$,
\[
    \OPT_{S,\Delta,\bot}^{+}(\Gamma)
    =\mathrm{BP}(\Gamma),
    \qquad
    \OPT_{S,\Delta}^{+}(\Gamma)
    =\mathrm{BP}_{+}(\Gamma).
\]
Both values are attained by equilibria of the
corresponding commitment extensions.
\end{proposition}

The construction above explains how to attain the
persuasion benchmarks. Conversely, any equilibrium
under a privately randomized sender commitment induces
a distribution over actions that satisfies the receiver's
incentive constraints. Its realized transfers can be
aggregated into compliance payments for a direct
persuasion scheme. Removing any realized burning
weakly increases the sender's payoff, giving the
matching upper bound. The formal proof appears in
Appendix~\ref{sec:proof-persuasion-equivalence}.

The proposition identifies optimal values, rather than
equating the two games or their equilibrium sets.
Persuasion commits directly to disclosure; our model
implements disclosure incentives through enforceable
contracts. The equivalence concerns the sender-only
optimistic problem with private randomization and
unrestricted burning. It does not characterize the
incentives created by bilateral commitment opportunities.

The equivalence also yields a polynomial-time algorithm
for arbitrary finite action sets. The persuasion
problems admit polynomial-size linear programs.
To obtain an efficient algorithm in our commitment
model, we must additionally represent the resulting
recommendation rule by a lottery with polynomially
many deterministic payment schemes.

\begin{corollary}\label{cor:general-randomized-computation}
For a sender-receiver game $\Gamma$ with rational
input data, $n$ states, and $m$ actions, an optimal
privately randomized sender commitment and a
supporting sequential equilibrium can be computed
in polynomial time.
This holds both when transfers and burning are
allowed and when only burning is allowed.

In either case, the commitment lottery can be chosen
to have support of size at most
\[
    n(m-1)+1.
\]
\end{corollary}

The persuasion equivalence reduces optimization to
a polynomial-size linear program. The remaining
issue is implementing its recommendation probabilities
without using an exponentially large lottery over
payment schemes. For each state, partition $[0,1]$
into intervals whose lengths are the probabilities
of recommending each action. A single random draw
from $[0,1]$ then determines a recommendation in
every state, and hence a deterministic action rule.
Together, the statewise partitions have at most
$n(m-1)$ interior endpoints, so their common
refinement yields at most $n(m-1)+1$ such rules.
Applying the contract construction from
Proposition~\ref{prop:persuasion-equivalence} to
each rule gives the required commitment lottery.
The full proof appears in
Appendix~\ref{app:general-randomized-computation}.

Thus, private randomization makes optimal sender
commitments computationally tractable for arbitrary
finite action sets. For binary actions,
Theorem~\ref{thm:opt-send-rand} gives a sharper result:
an explicit $O(n\log n)$ algorithm and an optimal
lottery over at most two schemes that differ at
only one state.

Private observation of the realized scheme is essential.
If the realization were publicly observed before
communication, each realization would induce a
deterministic continuation game. Its sender payoff
would be at most $\OPT_S^{+}(\Gamma)$, so averaging
over public realizations could not improve on the
optimal deterministic value.

\section{Conclusion}

We studied payment commitments as a mediating device
in sender-receiver games. By committing to transfers
or money burning, players can change the incentives
for communication without committing directly to
what information is disclosed. The benefits of this
device depend on who can commit, the order of
commitments, and whether equilibrium selection is
favorable.

For arbitrary finite action sets, we characterized
the pessimistic utility guarantees under unilateral
commitments. Sender-only commitments guarantee
exactly the welfare of the best constant action
in the worst equilibrium. This guarantee survives
when the receiver can commit after observing the
sender's commitment. Simultaneous commitments,
however, can introduce equilibria with lower welfare
than every equilibrium of the original game.

The comparison changes when we ask whether an
efficient equilibrium exists. Neither simultaneous
nor sender-first commitments always permit efficiency,
whereas receiver-first commitments do so for every
finite action set. With binary actions, we further
characterized the entire receiver-first equilibrium
payoff region. The characterization identifies how
state-dependent action choices permit off-path
penalties to support different divisions of surplus.
It also shows that receiver-first commitments can
sustain low-welfare equilibria despite admitting
efficient ones. Thus, protecting welfare in every
equilibrium and making efficiency attainable are
distinct benefits of a commitment protocol.

Our computational results reveal a further
distinction between deterministic and privately
randomized commitments. With binary actions,
optimal receiver commitments are computable in
linear time, while optimizing deterministic sender
commitments is NP-hard. Private randomization
admits an exact $O(n\log n)$ algorithm whose
implementation uses at most two schemes differing
at only one state, and yields an additive
approximation by deterministic commitments.
For arbitrary finite action sets, we connected
optimal privately randomized sender commitments
to Bayesian persuasion, with or without nonnegative
compliance payments. This equivalence also gives
polynomial-time algorithms with polynomial-size
commitment lotteries.

A natural next step is to study learning when
players do not know each other's utilities.
Repeated interactions provide observations of
recommendations, actions, and payments that may
allow players to adapt their commitments.
Which feedback permits low regret relative to
complete-information benchmarks, and whether
such learning preserves the benefits of particular
commitment orders, remain open questions.
Other directions include characterizing the
receiver-first payoff region beyond binary actions
and determining which guarantees survive bounds
on transfers and burning.

\section*{AI assistance}

The author used OpenAI's ChatGPT and Anthropic's
Claude during the development and preparation
of this paper. These tools assisted with drafting
and revising text from author-provided material,
developing detailed proofs from the author's
proof outlines, and working through technical
details. They also provided feedback on successive
drafts, helped identify gaps and errors, and
suggested extensions of the arguments, including
generalizations of some results from binary
to arbitrary finite action sets.
The author directed the research and takes
full responsibility for the results, proofs,
references, and final presentation.

\printbibliography

\appendix

\section{Proof of Lemma~\ref{lem:opt-eq}}\label{sec:proof-opt-eq}

Fix the committing player $i$ and write
$G:=\OPT_i^-(\Gamma)$ and $H:=\OPT_i^+(\Gamma)$.
In every equilibrium of the extension, player $i$ receives
at least $G$: otherwise, deviating to a scheme attaining $G$
would improve its payoff regardless of the continuation
selected after that deviation. Its equilibrium payoff is
at most $H$, by the definition of the optimistic optimum.

To attain $G$, select an $i$-worst continuation equilibrium
after every scheme. A scheme attaining $G$ is then an optimal
commitment under this continuation selection. Together, these
choices constitute an equilibrium of the extension with
payoff $G$. Selecting an $i$-best continuation after every
scheme and choosing a scheme attaining $H$ similarly gives
an equilibrium with payoff $H$. The same argument applies
when the available commitments are finitely supported
lotteries and the continuation games include the privately
observed realizations.

\section{Proof of Proposition~\ref{prop:binary-pareto}}\label{sec:proof-binary-pareto}

Fix a sequential equilibrium $e$. Write $q(m)$ for the
receiver's probability of choosing action $0$ after $m$, and
let $q_{\min}$ and $q_{\max}$ be its minimum and maximum over
all messages, including unused ones.

Suppose first that $q_{\min}<q_{\max}$.
Sender optimality and favorable tie-breaking imply that every
state with $a_S(\theta)=0$ sends only messages attaining
$q_{\max}$, except possibly states in which both players are
indifferent between actions. Similarly, every state with
$a_S(\theta)=1$ sends only messages attaining $q_{\min}$.
The exceptional states contribute zero to both players'
action-payoff differences and may be assigned to the
recommendation $0$ without affecting any incentive comparison.

At every on-path message attaining $q_{\max}>0$, action $0$
is a receiver best response. Summing the corresponding
unnormalized obedience inequalities gives
\[
\sum_{\theta:a_S(\theta)=0}p(\theta)
 \bigl(u_R(\theta,0)-u_R(\theta,1)\bigr)\geq0.
\]
The inequality also holds if this recommendation has zero
probability. Since $q_{\min}<1$, the same argument for the
minimum-probability messages gives obedience to recommendation
$1$. Thus the sender-preferred recommendation profile is
obedient. Receiver ties favor following the recommendation:
in every state assigned to a recommendation, its recommended
action weakly maximizes the sender's utility.

Replacing the response at each maximum-probability message
by action $0$, and at each minimum-probability message by
action $1$, does not change the receiver's expected payoff.
Indeed, a response changes only where the receiver originally
mixed and was therefore indifferent. Any remaining on-path
messages are used only by states in which both players are
indifferent and contribute no change either. Consequently,
$d_R(\Gamma)=u_R(e)$. The recommendation profile gives the
sender its statewise maximum payoff, so it weakly
Pareto-dominates $e$.

If $q_{\min}=q_{\max}$, the action distribution is the same
after every message. Averaging receiver optimality over
on-path messages shows that this distribution is supported
on prior-optimal actions. The babbling equilibrium selecting
a sender-best prior-optimal action therefore weakly
Pareto-dominates $e$.

Every equilibrium gives the receiver at least $b_R(\Gamma)$,
since it can ignore the message and use a prior-optimal
action. If $d_R(\Gamma)<b_R(\Gamma)$, the first case above
is impossible, establishing the first part of the proposition.

If $d_R(\Gamma)\geq b_R(\Gamma)$, obedience of the
sender-preferred recommendation profile follows directly from
\[
\begin{aligned}
\sum_{\theta:a_S(\theta)=1}p(\theta)
 \bigl(u_R(\theta,1)-u_R(\theta,0)\bigr)
 &=d_R(\Gamma)-R_0(\Gamma)\geq0,\\
\sum_{\theta:a_S(\theta)=0}p(\theta)
 \bigl(u_R(\theta,0)-u_R(\theta,1)\bigr)
 &=d_R(\Gamma)-R_1(\Gamma)\geq0.
\end{aligned}
\]
Sender optimality follows from the statewise choice of
$a_S(\theta)$, including the secondary tie-breaking criterion.
At unused messages, assign prior beliefs and a prior-optimal
response with favorable tie-breaking. These beliefs are
consistent with state-independent trembles to unused messages;
additional smaller-order trembles ensure complete mixing.
The sender cannot benefit from these responses because the
recommended action already maximizes its payoff, with ties
resolved in the receiver's favor.

This obedient equilibrium weakly Pareto-dominates babbling:
its sender payoff is a statewise maximum and its receiver
payoff is at least $b_R(\Gamma)$. Together with the preceding
case analysis, this proves the second part.

\section{Proof of Proposition~\ref{prop:rec-pessimistic}}\label{sec:proof-rec-pessimistic}

By committing to no payments, the receiver can guarantee
$b_R(\Gamma)$ in every continuation equilibrium: it can
ignore the sender's message and choose a best constant
action under the prior. Hence,
\[
    \OPT_R^{-}(\Gamma)\geq b_R(\Gamma).
\]

For the reverse inequality, fix any receiver payment
scheme $\tau_R$. The resulting communication game admits
a babbling equilibrium: the sender sends a fixed message,
and the receiver chooses a best constant action under
the prior, breaking ties in the sender's favor. Assign
prior beliefs and the same response after every unused
message. These beliefs are consistent with state-independent
trembles, and neither player benefits from deviating.

Since the sender makes no payments, the receiver's
expected utility in this equilibrium is
\[
    b_R(\Gamma^{(\mathbf{0},\tau_R)})
    =
    \max_{a\in A}
    \mathbb{E}_{\theta\sim p}
    \left[u_R(\theta,a)-\|\tau_R(\theta,a)\|\right]
    \leq b_R(\Gamma),
\]
where the inequality follows from nonnegativity of
payments. Thus, every receiver scheme admits a
continuation equilibrium giving the receiver at most
$b_R(\Gamma)$. Taking the supremum over schemes yields
\[
    \OPT_R^{-}(\Gamma)\leq b_R(\Gamma),
\]
which proves the claim.

\section{Proof of Proposition~\ref{prop:sen-pessimistic}}\label{sec:proof-sen-pessimistic}

Write
\[
    B:=W_{\mathrm{const}}(\Gamma)-b_R(\Gamma).
\]
We first prove the upper bound, then construct
a scheme guaranteeing $B$.

\paragraph{Upper bound.}
Fix any sender scheme $\tau_S$. The resulting game
admits a pure babbling equilibrium. In this equilibrium,
the receiver obtains at least $b_R(\Gamma)$ because
incoming transfers are nonnegative. Since the action
is constant, total welfare is at most
$W_{\mathrm{const}}(\Gamma)$: transfers preserve welfare,
and burning can only reduce it.
The sender therefore receives at most $B$.
Every sender scheme thus admits a continuation
equilibrium giving the sender at most $B$, proving
\[
    \OPT_S^{-}(\Gamma)\leq B.
\]

\paragraph{Construction.}
Choose
\[
    c\in\arg\max_{a\in A}W_a(\Gamma),
    \qquad
    t:=b_R(\Gamma)-R_c(\Gamma)\geq0.
\]
Let the sender transfer $t$ to the receiver on action
$c$ in every state, with no other transfers and no
burning on $c$.

Enumerate the other actions as $a_1,\ldots,a_{m-1}$,
where $m=|A|$, and choose
\[
    G:=
    1+
    \max\left(
        \{0\}\cup
        \{u_S(\theta,c)-t-u_S(\theta,a):
          \theta\in\Theta,\ a\neq c\}
    \right).
\]
Set $K_c:=0$ and $K_{a_j}:=jG$.
For each $a\neq c$, let the sender burn
\[
    \tau_S^\bot(\theta,a)
    :=
    u_S(\theta,a)-u_S(\theta,c)+t+K_a.
\]
These amounts are nonnegative by the choice of $G$.

The modified sender utilities are therefore
\[
    u_S^\tau(\theta,a)
    =
    u_S(\theta,c)-t-K_a
    \qquad\text{for every }a\in A.
\]
Thus the sender has the same strict ranking of actions
in every state, determined by the distinct constants
$K_a$, with $c$ ranked first.

\paragraph{Every continuation equilibrium plays $c$.}
Fix any sequential equilibrium under this scheme.
At every receiver information set, the sender's
conditional expected payoff from action $a$ is
\[
    \mathbb{E}[u_S(\theta,c)-t\mid m]-K_a.
\]
Consequently, favorable tie-breaking selects a unique
action among the receiver's best responses: the one
with the smallest $K_a$. Every receiver response
is therefore pure.

Because the sender's ranking of these responses
does not depend on the state, all on-path messages
induce the same action, say $a^*$. Averaging receiver
optimality over on-path messages shows that $a^*$
is optimal under the prior.

The receiver's constant-action payoff from $c$ is
\[
    R_c(\Gamma)+t=b_R(\Gamma),
\]
whereas its payoff from any other action is at most
$b_R(\Gamma)$. If $a^*\neq c$, both must therefore
give the receiver the same prior payoff.

At every on-path message, receiver optimality implies
that $a^*$ gives at least as much conditional utility
as $c$. Since the average difference is zero, the
receiver is indifferent between them after every
such message. But the sender strictly prefers $c$
in every state because $K_c<K_{a^*}$.
Favorable tie-breaking would therefore select $c$
instead of $a^*$, a contradiction.

Hence every continuation equilibrium plays $c$.
No burning is realized, and the sender receives
\[
    S_c(\Gamma)-t
    =
    W_c(\Gamma)-b_R(\Gamma)
    =B.
\]
This proves the lower bound and therefore equality.

Computing the constant-action benchmarks, $G$, and
the payment entries takes $O(n|A|)$ arithmetic
operations. The constructed amounts have polynomial
encoding length for rational input data.

\section{Proof of Theorem~\ref{thm:opt-rec}}\label{sec:proof-thm-opt-rec}

For each state $\theta$, define
\[
\Delta_S(\theta)
:=
|u_S(\theta,0)-u_S(\theta,1)|,
\qquad
\Delta_R(\theta)
:=
|u_R(\theta,0)-u_R(\theta,1)|.
\]

By Proposition~\ref{prop:binary-pareto}, for any receiver
payment scheme, every continuation equilibrium is weakly
Pareto-dominated by either a babbling equilibrium or a
sender-preferred obedient equilibrium. In a babbling
equilibrium, the receiver's expected utility cannot exceed
$b_R(\Gamma)$, since its own payments are nonnegative.
It therefore suffices to optimize the receiver's payoff
under the sender-preferred recommendation profile and
compare the resulting value with $b_R(\Gamma)$.

Fix a state $\theta$, and let $a_S(\theta)$ denote the
sender's preferred action in the original game, under
our tie-breaking convention. If a payment scheme leaves
this preferred action unchanged, the receiver's net
payoff at that state is at most
$u_R(\theta,a_S(\theta))$.

To make the other action sender-preferred, the receiver
must transfer at least $\Delta_S(\theta)$ to the sender
at that action. Thus, its net payoff from switching the
recommendation is at most
\[
u_R(\theta,1-a_S(\theta))-\Delta_S(\theta).
\]
Burning cannot improve either bound, since it reduces
the receiver's payoff without increasing the sender's
utility.

Consequently, switching the recommendation strictly
improves the receiver's statewise payoff precisely when
the players have opposite strict preferences and
$\Delta_R(\theta)>\Delta_S(\theta)$. In that case, a
transfer of exactly $\Delta_S(\theta)$ at the receiver's
preferred action attains the bound. The sender becomes
indifferent and selects that action because it gives
the receiver strictly higher modified utility.

Let $\tau_R^*$ be the scheme that makes this transfer
in every such state and makes no other payments.
By construction,
\[
d_R(\Gamma^{(\mathbf{0},\tau_R)})
\leq
d_R(\Gamma^{(\mathbf{0},\tau_R^*)})
\]
for every receiver scheme $\tau_R$.

Write
\[
D^*:=d_R(\Gamma^{(\mathbf{0},\tau_R^*)}).
\]
If $D^*>b_R(\Gamma)$, then
\[
D^*>b_R(\Gamma)
\geq b_R(\Gamma^{(\mathbf{0},\tau_R^*)}),
\]
where the second inequality follows from nonnegativity
of the receiver's payments. Proposition~\ref{prop:binary-pareto}
therefore guarantees that the sender-preferred
recommendation profile under $\tau_R^*$ can be completed
to a sequential equilibrium. This continuation gives
the receiver $D^*$.

If $D^*\leq b_R(\Gamma)$, select the zero payment scheme
and a babbling continuation, giving the receiver
$b_R(\Gamma)$.

The selected payoff is therefore
\[
\max\{b_R(\Gamma),D^*\}.
\]
The preceding bounds show that no receiver scheme and
continuation equilibrium can yield a higher payoff.
Complete the strategy profile by selecting a babbling
equilibrium after every other receiver scheme. The
selected commitment is then optimal, so this defines
an equilibrium of $\Gamma^R$ attaining $\OPT_R^+(\Gamma)$.

The scheme and both candidate values can be computed
in a single pass through the states. Hence the algorithm
uses $O(n)$ arithmetic operations. For rational input data, the algorithm also has
polynomial bit complexity. Its payment amounts are
differences of input utilities, and its payoff
comparisons use sums of products of these amounts
with prior probabilities. All resulting rational
numbers therefore have encoding length polynomial
in the input length.

\section{Proof of Theorem~\ref{thm:opt-send-hard}}\label{sec:proof-opt-send-hard}

We prove the theorem by reducing from the NP-complete decision
version of $0$--$1$ Knapsack with integer data
\cite{karp1972reducibility}. An instance consists of positive
integer weights $w_1,\ldots,w_n$, positive integer values
$v_1,\ldots,v_n$, a capacity $W$, and a target value $V$.
The question is whether there exists a subset
$J\subseteq\{1,\ldots,n\}$ such that
\[
\sum_{i\in J}w_i\leq W
\qquad\text{and}\qquad
\sum_{i\in J}v_i\geq V.
\]
All integers are encoded in binary. Write
\[
B:=\sum_{i=1}^n w_i,
\qquad
T:=\sum_{i=1}^n v_i.
\]
We may restrict attention to instances with
$1\leq W\leq B$ and $1\leq V\leq T$, since instances outside
these ranges can be decided directly.

\paragraph{Construction and intuition.}

Given a Knapsack instance $K$, construct a sender-receiver
game $\Gamma_K$ with actions $\{0,1\}$ and $n+2$ equally
likely states
\[
\Theta=\{\theta_1,\ldots,\theta_n,L,R\}.
\]
Let
\[
N:=n+2,
\qquad
C_0=C_1:=B+1,
\qquad
\epsilon:=\frac12.
\]
The payoffs, listed in the order $(u_S,u_R)$, are
\[
\begin{array}{c|cc}
 & 0 & 1\\
\hline
L & (0,C_0) & (0,-C_0)\\
\theta_i & (0,C_1v_i) & (w_i,0)\\
R & (0,0) & (\epsilon,C_1(T-V)).
\end{array}
\]

Initially, the sender prefers action $1$ in every item state
and in $R$, and recommends action $0$ in $L$ by favorable
tie-breaking. The receiver's prior-optimal action is $0$:
its prior payoff advantage over action $1$ is
\[
\frac{2C_0+C_1V}{N}>0.
\]
The sender therefore receives zero in the babbling equilibrium.

Under the sender-preferred recommendation profile, the
receiver instead obtains
\[
d_R(\Gamma_K)=\frac{C_0+C_1(T-V)}{N},
\]
whereas
\[
b_R(\Gamma_K)=\frac{C_0+C_1T}{N}.
\]
Thus the recommendation profile fails receiver obedience
by a gap of $C_1V/N$.

To close this gap, the sender can select an item $i$ by
committing to transfer $w_i$ to the receiver at
$(\theta_i,1)$. This makes the sender indifferent between
the two actions in $\theta_i$. Since $C_1v_i>w_i$, favorable
tie-breaking then makes the sender recommend action $0$.
The promised transfer is not executed. Instead, the sender
forgoes utility $w_i$, while the receiver gains utility
$C_1v_i$ at that state.

Selecting a subset $J$ therefore closes the obedience gap
whenever $\sum_{i\in J}v_i\geq V$, at an expected cost to the
sender of $\sum_{i\in J}w_i/N$. This suggests choosing a
minimum-weight subset that reaches the target value.

There are two possible ways to circumvent this construction.
The sender might try to make action $1$ optimal in a babbling
continuation, or use smaller on-path payments to close the
obedience gap without selecting enough items. The constants
$C_0$ and $C_1$ rule out both possibilities. Switching the
babbling action requires payments exceeding all available
sender utility. Moreover, if a selected subset falls short
of $V$, integrality leaves an obedience gap of at least
$C_1$, before normalization by $N$. Closing that gap through
on-path payments would again exhaust all available sender
utility.

We formalize these observations below.

\begin{proposition}\label{prop:knapsack-reduction}
For the game $\Gamma_K$ constructed above,
\[
\OPT_S^+(\Gamma_K)
=
\frac{B-\kappa+\epsilon}{N},
\]
where
\[
\kappa:=
\min\left\{
\sum_{i\in J}w_i:
J\subseteq\{1,\ldots,n\},
\ \sum_{i\in J}v_i\geq V
\right\}.
\]
This optimum is attained by an equilibrium of $\Gamma_K^S$.
Consequently, $K$ is a yes-instance if and only if
\[
\OPT_S^+(\Gamma_K)
\geq
\frac{B-W+\epsilon}{N}.
\]
\end{proposition}

\begin{proof}
We first bound babbling payoffs, then construct a scheme
for each subset reaching $V$, and finally prove that no
other scheme gives the sender a higher payoff.

\medskip
\noindent
\textbf{Babbling continuations give the sender at most zero.}

Fix any sender scheme $\tau_S$, allowing both transfers
and burning. For $a\in\{0,1\}$, write
\[
t_a:=\sum_{\theta\in\Theta}\tau_S^R(\theta,a),
\qquad
c_a:=\sum_{\theta\in\Theta}\|\tau_S(\theta,a)\|.
\]
Thus $c_a\geq t_a\geq0$.

If action $0$ is played independently of the state, the
sender's expected payoff is $-c_0/N\leq0$.

If action $1$ is played with positive probability in a
babbling equilibrium, it must be prior-optimal for the
receiver. Hence
\[
-C_0+C_1(T-V)+t_1
\geq
C_0+C_1T+t_0,
\]
which implies
\[
t_1-t_0\geq2C_0+C_1V.
\]
The sender's expected payoff from always playing $1$ is
therefore at most
\[
\frac{B+\epsilon-c_1}{N}
\leq
\frac{B+\epsilon-2C_0-C_1V}{N}
<0.
\]
A mixture of the two prior-optimal actions also gives the
sender at most zero.

\medskip
\noindent
\textbf{Every subset reaching $V$ yields the claimed payoff.}

Fix $J\subseteq\{1,\ldots,n\}$ with
$\sum_{i\in J}v_i\geq V$. Define $\tau_S^J$ by
\[
(\tau_S^J)^R(\theta_i,1)=w_i
\qquad\text{for }i\in J,
\]
with all other payments equal to zero.

Under this scheme, the sender recommends $0$ in $L$ and
in the item states indexed by $J$, and recommends $1$
in $R$ and in the remaining item states. In a selected
item state, this follows from sender indifference and
\[
C_1v_i>w_i.
\]

The receiver's obedience inequality for recommendation $1$,
multiplied by $N$, is
\[
C_1(T-V)-C_1\sum_{i\notin J}v_i
=
C_1\left(\sum_{i\in J}v_i-V\right)
\geq0.
\]
If equality holds, receiver tie-breaking selects $1$,
since every sender type making this recommendation
strictly prefers $1$.

For recommendation $0$, the corresponding inequality is
\[
2C_0+\sum_{i\in J}(C_1v_i-w_i)>0.
\]
Thus both recommendations are obeyed.

At unused messages, assign prior beliefs and action $0$.
This response is strictly optimal under the prior, since
\[
2C_0+C_1V-\sum_{i\in J}w_i>0.
\]
These beliefs are consistent with uniform trembles, and
no sender type benefits from using an unused message.
The profile therefore defines a sequential equilibrium
of the continuation game.

All promised transfers are off path, so the sender's
expected payoff is
\[
\frac{B-\sum_{i\in J}w_i+\epsilon}{N}.
\]
In particular, a subset attaining $\kappa$ yields
$(B-\kappa+\epsilon)/N>0$.

\medskip
\noindent
\textbf{No arbitrary scheme can do better.}

Fix any sender scheme and any continuation equilibrium
giving the sender strictly positive utility. By
Proposition~\ref{prop:binary-pareto} and the babbling bound
above, this equilibrium is weakly Pareto-dominated by a
sender-preferred obedient equilibrium of the same
continuation game. It therefore suffices to bound such
obedient equilibria.

Let $a(\theta)$ be the action recommended and played in
this obedient equilibrium, and define
\[
J:=\{i:a(\theta_i)=0\}.
\]
Also let
\[
\eta_R:=\mathbf{1}_{\{a(R)=1\}},
\qquad
\eta_L:=\mathbf{1}_{\{a(L)=1\}},
\qquad
\Pi:=\sum_{\theta\in\Theta}
       \|\tau_S(\theta,a(\theta))\|.
\]
Here $\Pi/N$ is the sender's expected executed payment,
including burning.

The sender's expected utility satisfies
\[
NU_S
=
B-\sum_{i\in J}w_i+\eta_R\epsilon-\Pi.
\]
Since $U_S>0$, it follows that
\[
\Pi<B+\epsilon<C_1.
\]

Receiver obedience after recommendation $1$ gives
\[
0\leq
\sum_{\theta:a(\theta)=1}
\left(
u_R(\theta,1)-u_R(\theta,0)
+\tau_S^R(\theta,1)-\tau_S^R(\theta,0)
\right).
\]
The sum of the original utility differences in this
expression is
\[
-C_1\sum_{i\notin J}v_i
+\eta_R C_1(T-V)-2\eta_L C_0
\leq
C_1\left(\sum_{i\in J}v_i-V\right),
\]
where we use $T-V\geq0$. Moreover, nonnegativity of
transfers implies
\[
\sum_{\theta:a(\theta)=1}
\left(
\tau_S^R(\theta,1)-\tau_S^R(\theta,0)
\right)
\leq
\sum_{\theta:a(\theta)=1}\tau_S^R(\theta,1)
\leq \Pi.
\]
Consequently,
\[
C_1\left(V-\sum_{i\in J}v_i\right)
\leq\Pi<C_1.
\]
Since $V$ and the $v_i$ are integers, this forces
\[
\sum_{i\in J}v_i\geq V.
\]
Thus $\sum_{i\in J}w_i\geq\kappa$, and hence
\[
U_S
\leq
\frac{B-\sum_{i\in J}w_i+\epsilon}{N}
\leq
\frac{B-\kappa+\epsilon}{N}.
\]
Nonpositive-payoff equilibria also satisfy this bound,
because $\kappa\leq B$.

Combining the upper bound with the construction above
proves the formula for $\OPT_S^+(\Gamma_K)$.

Finally, select the constructed obedient continuation
after a scheme associated with a subset attaining
$\kappa$, and a babbling continuation after every other
sender scheme. Every commitment deviation gives the
sender at most zero, whereas the selected scheme gives
strictly positive utility. These commitments and
continuations therefore define an equilibrium of
$\Gamma_K^S$ attaining the optimum.

The Knapsack instance is feasible exactly when
$\kappa\leq W$, which is equivalent to the stated
payoff inequality.
\end{proof}

The construction uses $n+2$ states, and all its rational
parameters have encoding length polynomial in the
encoding length of $K$. Computing the optimal sender
payoff, or an equilibrium attaining it, would therefore
decide $0$--$1$ Knapsack in polynomial time. This proves
Theorem~\ref{thm:opt-send-hard}.

\section{Proof of Theorem~\ref{thm:opt-send-rand}}\label{sec:proof-thm-send-rand}

Randomized commitments allow the sender to make a recommendation
change in a state with any desired probability. This makes the
corresponding utility gains and losses divisible and leads to
a greedy algorithm.

The running time counts arithmetic operations. As in the
definition of $\Gamma^S_\Delta$, the sender publicly commits
to a lottery over payment schemes, but its realization is
observed only by the sender before communication.

\paragraph{Intuition.}
For a commitment $\tau_S$, write
\[
    R_a(\Gamma^{\tau_S})
    :=
    \mathbb{E}[u_R^{\tau_S}(\theta,a)]
\]
for the receiver's expected utility from always playing
action $a$. For a randomized commitment, the expectation
also includes its realization. Let
$d_R(\Gamma^{\tau_S})$ be the receiver's expected utility
under the sender-preferred recommendation profile.

With binary actions, this profile can be completed to a
sequential equilibrium if and only if
\[
    d_R(\Gamma^{\tau_S})
    \geq R_a(\Gamma^{\tau_S})
    \qquad\text{for both }a\in\{0,1\}.
\]
Indeed, obedience following recommendation $0$ corresponds
to the comparison with always playing $1$, and obedience
following recommendation $1$ corresponds to the comparison
with always playing $0$.

Start with no payments and the sender's preferred original
action in every state. If both inequalities already hold,
this is optimal: the sender receives her largest possible
original payoff and makes no payments. Otherwise, the
sender must close the remaining obedience gaps.

There are two useful ways to do so. First, consider a
state $\theta$ in which the players have strictly opposing
preferences, and write
\[
    \Delta_i(\theta)
    :=
    |u_i(\theta,0)-u_i(\theta,1)|
    \qquad(i\in\{S,R\}).
\]
The sender can commit to burn $\Delta_S(\theta)$ on her
originally preferred action. She then becomes indifferent
between the two actions and, by tie-breaking in favor of
the receiver, recommends the other action.
This decreases her expected utility by
$p(\theta)\Delta_S(\theta)$ and increases the receiver's
recommended payoff by $p(\theta)\Delta_R(\theta)$.
The burn is never realized, and neither constant-action
payoff changes.

The efficiency of this exchange is therefore
$\Delta_R(\theta)/\Delta_S(\theta)$.
By randomizing the commitment, the sender can obtain any
fraction of this exchange.

If the sender instead promises to transfer
$\Delta_S(\theta)$ to the receiver on her originally
preferred action, the receiver's recommended payoff
increases by
\[
    p(\theta)
    \max\{\Delta_R(\theta),\Delta_S(\theta)\},
\]
but the constant-action payoff for that action also
increases by $p(\theta)\Delta_S(\theta)$.
In particular, when
$\Delta_R(\theta)\geq\Delta_S(\theta)$, burning gives
the same recommended payoffs with weakly easier obedience
constraints.

The second useful operation is to transfer utility to the
receiver on recommended actions. An aggregate expected
transfer $t$ on outcomes recommending action $a$ increases
the receiver's recommended payoff and the constant-action
payoff for $a$ by $t$, leaving the other constant-action
payoff unchanged. The sender loses $t$ in expected utility.
If necessary, an additional off-path burn on the other
action preserves the sender's preference for the
recommended action. This additional burn is not realized.

To compare these operations, suppose first that the two
constant-action payoffs coincide at $R$, while the
recommended payoff is $d<R$. Let $D:=R-d$.
A transfer of $D$ on each recommendation makes obedience
possible: the recommended payoff increases by $2D$,
while each constant-action payoff increases by $D$.
The sender pays $2D$ to close a common gap of $D$.
Thus, paired transfers have efficiency $1/2$.
Burning is preferable whenever
\[
    \frac{\Delta_R(\theta)}{\Delta_S(\theta)}
    >\frac12.
\]

Now suppose that the constant-action payoffs differ,
say $R_0>R_1$, and that $d<R_0$.
A transfer $t$ on recommendations of action $1$ changes
these quantities to
\[
    d+t,\qquad R_0,\qquad R_1+t.
\]
It closes the larger gap at efficiency $1$.
We can continue until either obedience is achieved or
the two constant-action payoffs become equal.
After equality, the preceding comparison with paired
transfers applies.

These observations suggest the following order.
First use burning opportunities with efficiency at least
$1$. Next, if the constant-action payoffs differ, transfer
on recommendations of the worse constant action until
obedience is achieved or the two benchmarks coincide.
Then use the remaining burning opportunities with
efficiency above $1/2$. Finally, close any remaining gap
with paired transfers.

\paragraph{The algorithm.}
Let $a_S(\theta)$ be the sender's preferred original action,
with ties resolved as specified in the model. Initialize
\[
    d=\sum_{\theta}p(\theta)
        u_R(\theta,a_S(\theta)),
    \qquad
    B_a=\sum_{\theta}p(\theta)u_R(\theta,a),
    \qquad
    T_a=0.
\]
Here, $d$ records the receiver's recommended payoff,
including the transfers planned so far; $B_a$ records
the corresponding constant-action payoff; and $T_a$
records the aggregate expected transfer planned on
recommendations of action $a$.

Let $I$ be the set of states with strictly opposing
preferences. For $\theta\in I$, define
\[
    g_\theta=p(\theta)\Delta_R(\theta),
    \qquad
    c_\theta=p(\theta)\Delta_S(\theta),
    \qquad
    r_\theta=\frac{g_\theta}{c_\theta}.
\]
The quantities $g_\theta$ and $c_\theta$ are the receiver's
gain and the sender's cost from switching the action in
that state with probability one.
Sort these states in decreasing order of $r_\theta$,
and initially set their switching probabilities
$x_\theta$ to zero.

Whenever the algorithm increases $x_\theta$ by $\delta$,
it increases $d$ by $\delta g_\theta$.
Whenever it adds an aggregate transfer $t$ on
recommendation $a$, it increases $T_a$, $d$, and $B_a$
by $t$.
The allocation of the aggregate transfers to particular
outcomes is deferred until the final construction.

If $d\geq\max\{B_0,B_1\}$ initially, return the
no-payment sender-preferred equilibrium.
Otherwise, proceed as follows.
\begin{enumerate}
    \item
    Process the states with $r_\theta\geq1$ in sorted
    order. For each such state, increase its switching
    probability by
    \[
        \delta
        =
        \min\left\{
            1,\,
            \frac{\max\{B_0,B_1\}-d}{g_\theta}
        \right\}.
    \]
    If $d=\max\{B_0,B_1\}$, stop the allocation procedure.

    \item
    If the benchmarks differ, let $a$ be the action
    with the smaller benchmark and add the transfer
    \[
        t=\min\{B_{1-a}-d,\ B_{1-a}-B_a\}
    \]
    on recommendation $a$.
    If $d=\max\{B_0,B_1\}$, stop the allocation procedure.
    Otherwise, the two benchmarks now coincide.

    \item
    Process the remaining states with
    $1/2<r_\theta<1$ in sorted order.
    For each state, increase its switching probability by
    \[
        \delta
        =
        \min\left\{
            1,\,
            \frac{B_0-d}{g_\theta}
        \right\}.
    \]
    If $d=B_0=B_1$, stop the allocation procedure.

    \item
    If a common gap $D:=B_0-d=B_1-d$ remains,
    add an aggregate transfer $D$ on each recommendation.
\end{enumerate}

The output consists of the resulting switching probabilities
and aggregate transfers. The construction in the next
subsection converts them into a lottery over at most two
payment schemes and an obedient continuation equilibrium.

Each state is processed at most once, and all payoffs are
updated incrementally. Thus, the allocation procedure
takes linear time after sorting. The following proof
establishes its optimality and gives the implementation.

\subsection{Correctness of the Algorithm}

The proof has three parts. We first identify the minimum
compensation required by an allocation. We then show that
the algorithm selects an allocation maximizing the sender's
payoff after this compensation. Finally, we implement its
output with a randomized commitment.

\paragraph{Necessary compensation.}
In this proof, write
\[
    R_a:=\sum_{\theta}p(\theta)u_R(\theta,a)
\]
for the original constant-action payoffs, excluding
transfers. For any allocation, let $v$ denote the
receiver's expected original payoff, also excluding
transfers. Define
\[
    h(v):=(R_0-v)_+ +(R_1-v)_+,
    \qquad z_+:=\max\{z,0\}.
\]

This expression is the binary-action specialization
of the minimum nonnegative compliance payments
characterized by Dughmi et al.~\cite{dughmi2019persuasion}.
Below, we first show that the same expression bounds
the transfers required by any continuation equilibrium
in our commitment model. We then construct contracts
attaining this bound.

Consider any randomized commitment and continuation
equilibrium, and let $A$ be the action played. For
$a\in\{0,1\}$, let
\[
    P_a
    :=
    \mathbb{E}\!\left[
        \mathbf{1}_{\{A=a\}}\tau_S^R(\theta,a)
    \right].
\]
Thus, $P_a$ is the unconditional expected transfer
actually paid on outcomes where action $a$ is played.

Aggregate the receiver's optimality inequalities for
action $a$ over messages, weighting each inequality by
the probability that the message occurs and that action
$a$ is played. This gives
\[
\begin{aligned}
    0
    &\leq
    \mathbb{E}\!\left[
        \mathbf{1}_{\{A=a\}}
        \left(
            u_R(\theta,a)-u_R(\theta,1-a)
            +\tau_S^R(\theta,a)
            -\tau_S^R(\theta,1-a)
        \right)
    \right]\\
    &\leq v-R_{1-a}+P_a.
\end{aligned}
\]
The second inequality follows from nonnegativity of
transfers and the identity
\[
    v-R_{1-a}
    =
    \mathbb{E}\!\left[
        \mathbf{1}_{\{A=a\}}
        \bigl(u_R(\theta,a)-u_R(\theta,1-a)\bigr)
    \right].
\]
Therefore,
\[
    P_a\geq(R_{1-a}-v)_+,
    \qquad
    P_0+P_1\geq h(v).
\]
Any realized burning further reduces the sender's payoff.
Consequently, the sender's utility is bounded above by
her original payoff from the allocation minus $h(v)$.

\paragraph{Optimality of the switching rule.}
Define
\[
    F_S:=\sum_{\theta}p(\theta)
        u_S(\theta,a_S(\theta)),
    \qquad
    v^0:=\sum_{\theta}p(\theta)
        u_R(\theta,a_S(\theta)).
\]
Outside $I$, replacing any action by $a_S(\theta)$ weakly
increases both players' original utilities. Since $h$
is nonincreasing, this replacement cannot decrease
the preceding upper bound.

It therefore suffices to optimize over switching
probabilities $x_\theta\in[0,1]$ for states in $I$.
The resulting upper bound is
\[
    F_S-\sum_{\theta\in I}c_\theta x_\theta
    -
    h\left(v^0+\sum_{\theta\in I}g_\theta x_\theta\right).
    \tag{$\star$}
\]

The bound also covers babbling equilibria and commitments
that induce a different constant action. Indeed, writing
$S_a:=\mathbb{E}[u_S(\theta,a)]$, a constant-$a$ allocation
has receiver payoff $R_a$ and therefore gives the sender
at most
\[
    S_a-h(R_a)
    =
    S_a-\bigl(\max\{R_0,R_1\}-R_a\bigr).
\]
This includes the original babbling payoff when $a$ is
a receiver-optimal constant action. Although a constant
allocation may prescribe a non-sender-preferred action
outside $I$, the replacement argument above weakly
improves its objective value. Thus, optimizing
$(\star)$ already accounts for both babbling and buying
the other action; no separate comparison is required.

For any fixed increase in the receiver's original payoff,
the sender's switching cost is minimized by processing
states in decreasing order of $g_\theta/c_\theta$.
Indeed, if a more expensive state is used while a cheaper
one remains unfilled, moving an equal amount of receiver
gain to the cheaper state weakly reduces the cost.

Along this sorted order, the marginal saving in
compensation per unit increase in the receiver's original
payoff is
\[
    \begin{cases}
        2,&v<\min\{R_0,R_1\},\\
        1,&\min\{R_0,R_1\}<v<\max\{R_0,R_1\},\\
        0,&v>\max\{R_0,R_1\}.
    \end{cases}
\]
At the benchmarks, the corresponding one-sided rates apply.
Thus, switches with ratio at least $1$ can be taken
until the larger benchmark is reached. Among the remaining
states, a switch is beneficial precisely when its ratio
exceeds $1/2$ and $v$ is still below the smaller benchmark.
Once either condition fails, no later switch can improve
$(\star)$.

These are exactly the algorithm's switching decisions.
To see this for its third step, suppose without loss
of generality that $R_0\geq R_1$, and let $v$ be the
original receiver payoff after the first step.
If $v\geq R_1$, the second step achieves obedience and
the algorithm stops. Otherwise, the second step transfers
$R_0-R_1$ on recommendation $1$, giving
\[
    d=v+R_0-R_1,
    \qquad B_0=B_1=R_0.
\]
Hence, the remaining common gap is $R_1-v$.
The third step therefore switches states precisely until
the original payoff reaches $R_1$, or until the ratios
above $1/2$ are exhausted.

It follows that the algorithm maximizes $(\star)$.
Its aggregate transfers also satisfy
\[
    T_a=(R_{1-a}-v^*)_+,
\]
where $v^*$ is the final receiver payoff before transfers.
Indeed, if only the larger gap is closed by transfers,
this follows directly from step~2. If both gaps remain,
the initial equalizing transfer and the final paired
transfers give the same formula. Consequently, the total
planned transfer is exactly $h(v^*)$.

\paragraph{A weighted-utility interpretation.}
The optimal allocation also admits a simple interpretation:
there exists $\lambda\in[0,2]$ such that every action
played with positive probability in state $\theta$
maximizes
\[
    u_S(\theta,a)+\lambda u_R(\theta,a).
\]
These are the original utilities, before payments.
In a disagreeing state, switching from the sender's
preferred action to the receiver's preferred action
improves this weighted sum precisely when
\[
    \lambda\Delta_R(\theta)>\Delta_S(\theta).
\]
The algorithm's ordering therefore selects states
according to the threshold
$\Delta_S(\theta)/\Delta_R(\theta)$.
At a partially switched state, the two actions tie
under the selected weight.

If the initial sender-preferred allocation is obedient,
we may take $\lambda=0$. If the allocation procedure
terminates during its first burning phase, a threshold
weight in $(0,1]$ supports the allocation. If it
terminates in the equalizing-transfer phase, we may
take $\lambda=1$. Termination during the second
burning phase admits a threshold weight in $(1,2)$,
and use of the final paired transfers permits
$\lambda=2$.

Thus nonnegative payments lead to an endogenous weight
between zero and two on receiver utility. The endpoint
$\lambda=2$ coincides with the two-action allocation
rule obtained by Dughmi et al.~\cite{dughmi2019persuasion}
when payments of either sign are permitted.

\paragraph{Implementation.}
The algorithm fully switches a prefix of the sorted
states and partially switches at most one additional
state. Its allocation can therefore be represented
by a lottery over at most two deterministic action maps
\[
    a^j:\Theta\longrightarrow\{0,1\}.
\]
If a state is partially switched, the two maps differ
only at that state, and their probabilities implement
its switching probability.

Let $\alpha_a$ be the unconditional probability that
the selected map prescribes action $a$, and set
\[
    z_a:=
    \begin{cases}
        T_a/\alpha_a,&\alpha_a>0,\\
        0,&\alpha_a=0.
    \end{cases}
\]
If $\alpha_a=0$, the other action is always prescribed,
so $v^*=R_{1-a}$ and $T_a=0$.

For each map $a^j$, construct a payment scheme as follows.
At state $\theta$, write $a=a^j(\theta)$.
On action $a$, the sender transfers $z_a$ to the receiver.
On the other action, she transfers nothing to the receiver
and burns
\[
    1+\max\left\{
        0,\,
        u_S(\theta,1-a)-u_S(\theta,a)+z_a
    \right\}.
\]
There are no other payments.

The sender privately observes the realized scheme and
recommends its prescribed action. The off-path burn makes
this action strictly preferable for her in every state
and realization.

For a recommendation $a$ occurring with positive
probability, the receiver's unconditional gain from
following it rather than switching is
\[
\begin{aligned}
    &\mathbb{E}\!\left[
        \mathbf{1}_{\{a^j(\theta)=a\}}
        \bigl(u_R(\theta,a)-u_R(\theta,1-a)+z_a\bigr)
    \right]\\
    &\qquad =
        v^*-R_{1-a}+\alpha_a z_a\\
    &\qquad =
        v^*-R_{1-a}+(R_{1-a}-v^*)_+
        \geq 0.
\end{aligned}
\]
Following the recommendation is therefore optimal.
When the receiver is indifferent, following also respects
tie-breaking in favor of the sender, since every compatible
sender type strictly prefers the prescribed action.

If a recommendation is unused, assign it the prior belief
over states and scheme realizations and let the receiver
best respond. These beliefs are consistent with trembles
that send the unused recommendation with the same vanishing
probability at every sender type. The sender cannot benefit
from deviating, since her prescribed action is strictly
preferred to the other action.

This gives a sequential equilibrium of the continuation
game. No burning is realized, and the total expected
transfer is
\[
    \alpha_0z_0+\alpha_1z_1
    =T_0+T_1=h(v^*).
\]
The sender therefore attains the maximum of $(\star)$,
which upper-bounds her utility under every randomized
commitment and continuation equilibrium.

Finally, specify a babbling continuation equilibrium after
every other announced commitment. Since the constructed
commitment attains the global upper bound, the sender
has no profitable commitment-stage deviation. The resulting
profile is a sequential equilibrium of the extended game.

Computing the initial quantities takes $O(n)$ operations,
sorting takes $O(n\log n)$, and the subsequent scan and
construction of at most two payment schemes take $O(n)$.
The algorithm also has polynomial bit complexity.
The quantities $c_\theta$, $g_\theta$, and the
constant-action benchmarks are formed by differences,
products, and sums of input rationals. Sorting the
ratios uses cross-multiplication. The algorithm
partially switches at most one state; its switching
probability is obtained by dividing the remaining
payoff gap by that state's contribution. All preceding
updates are sums of contributions from fully switched
states. Finally, the implementation computes the
probabilities $\alpha_a$, the transfers
$z_a=T_a/\alpha_a$ when $\alpha_a>0$, and the off-path
burns using a constant number of further arithmetic
operations per state. Consequently, every intermediate
quantity and output parameter has encoding length
polynomial in the input length.
This proves the theorem.

\subsection{Proof of Corollary~\ref{cor:randomization-gap}}\label{sec:cor-random}

The randomized algorithm partially switches at most
one state. Round that switch upward to probability one.
The sender's original expected payoff decreases by
at most $p(\theta)\Delta_S(\theta)$, while the receiver's
original expected payoff increases. Since the required
compensation $h$ is nonincreasing, implementing the
rounded allocation loses at most this amount.

If the resulting sender payoff is at least
$W_{\mathrm{const}}(\Gamma)-b_R(\Gamma)$, babbling
continuations after all other commitments support
the scheme as an equilibrium of $\Gamma^S$.
Otherwise, use the optimal babbling construction
of Proposition~\ref{prop:sen-pessimistic}, which gives
the sender a higher payoff.
The additional computation takes $O(n)$ operations.

\section{Proof of Proposition~\ref{prop:simultaneous-welfare}}\label{sec:proof-worst-welfare}

\begin{proof}
Consider the game $\Gamma$ with equally likely states
$\Theta=\{L,R\}$, actions $A=\{0,1\}$, and the following
payoffs, written in the order $(u_S,u_R)$:
\[
\begin{array}{c|cc}
    &0&1\\ \hline
    L&(4,4)&(0,0)\\
    R&(4,0)&(0,2).
\end{array}
\]

We first show that every equilibrium of the original game
has welfare $6$. Since the sender strictly prefers action
$0$ in both states, every message sent with positive
probability must induce the same maximal probability $q$
of action $0$. The receiver's expected utility is therefore
\[
    2q+(1-q)=1+q.
\]
The receiver can guarantee utility $2$ by ignoring the
message and always playing $0$, so equilibrium requires
$q=1$. Thus, every original equilibrium gives utilities
$(4,2)$ and welfare $6$.

Moreover, the maximum original welfare is $8$ in state
$L$ and $4$ in state $R$. Hence,
\[
    W^*(\Gamma)=\frac12(8+4)=6.
\]
Transfers preserve total utility and burning reduces it,
so no equilibrium with commitments can have welfare
greater than $6$.

We next construct a simultaneous-commitment equilibrium
with welfare $5$. The sender commits to burn $5$ at
$(R,0)$, and the receiver commits to burn $12$ at $(R,0)$.
All other payments are zero. The resulting payoffs are
\[
\begin{array}{c|cc}
    &0&1\\ \hline
    L&(4,4)&(0,0)\\
    R&(-1,-12)&(0,2).
\end{array}
\]
At this commitment profile, prescribe truthful state
disclosure followed by action $0$ in state $L$ and
action $1$ in state $R$. Both players strictly prefer
the prescribed action in each state, so this is a
sequential equilibrium of the continuation game.
No burning is realized, and the expected utilities
are $(2,3)$, giving welfare $5$.

After every other commitment profile, prescribe a
babbling continuation equilibrium in which the receiver
plays a pure best response under the prior, with the
specified tie-breaking convention. Such a continuation
exists after every commitment profile: the sender's
message is independent of the state, and beliefs after
unused messages can also be set equal to the prior.

Consider a unilateral deviation by the receiver.
The sender's commitment provides no transfers to the
receiver, and the prescribed continuation is babbling.
The receiver's original expected payoff from either
constant action is at most $2$, and its own payments
can only reduce that payoff. Every different commitment
therefore gives the receiver at most $2$, strictly less
than its equilibrium payoff of $3$.

Now consider a unilateral deviation by the sender.
With the receiver's commitment fixed, its expected
payoffs before incoming transfers are $-4$ from constant
action $0$ and $1$ from constant action $1$.
Let
\[
    T_a:=\mathbb{E}_{\theta\sim p}
    [\tau_S^R(\theta,a)]
\]
be the expected transfer promised to the receiver on
constant action $a$ under the deviating commitment.
If the babbling continuation plays $0$, receiver
optimality requires
\[
    -4+T_0\geq1+T_1,
\]
so $T_0\geq5$. The sender's expected utility is then
at most $4-T_0\leq-1$.
If the continuation plays $1$, the sender's original
payoff is zero, so its utility after payments is at
most zero. Thus, every different commitment gives the
sender at most $0$, strictly less than its equilibrium
payoff of $2$.

Both commitments are therefore strict best responses,
and the constructed profile is a sequential equilibrium
of $\Gamma^{SR}$. Consequently,
\[
    \mathrm{EQ}_{SR,w}^{-}(\Gamma)\leq5.
\]

Finally, we construct a simultaneous-commitment equilibrium
with welfare $6$. Both players commit to no payments,
and a babbling continuation is prescribed after every
commitment profile. On the equilibrium path, the receiver
plays $0$, yielding utilities $(4,2)$.

A receiver deviation gives at most $2$, because the
continuation is babbling, the sender makes no transfers,
and receiver payments only reduce its utility.
Following a sender deviation, the receiver can still
guarantee at least $2$ by playing constant action $0$:
the sender's incoming transfers are nonnegative.
Since total welfare is at most $6$, the sender's utility
is at most $4$. Neither player can therefore profitably
deviate.

This equilibrium attains the welfare upper bound, and hence
\[
    \mathrm{EQ}_{SR,w}^{-}(\Gamma)
    \leq5
    <6
    =
    \mathrm{EQ}_{SR,w}^{+}(\Gamma)
    =
    W^*(\Gamma).
\]
\end{proof}

\section{Proof of Theorem~\ref{thm:sender-first-guarantees}}\label{sec:proof-sender-first-guarantees}

Write
\[
    S_a=\sum_{\theta}p(\theta)u_S(\theta,a),
    \qquad
    R_a=\sum_{\theta}p(\theta)u_R(\theta,a),
    \qquad
    W_a=S_a+R_a.
\]
We first prove the sender's guarantee. The case
$|A|=1$ is immediate, so suppose that $|A|\ge2$.

Fix $\varepsilon>0$, choose
$c\in\arg\max_{a\in A}W_a$, and set
\[
    s=b_R(\Gamma)-R_c+\varepsilon.
\]
The sender commits to transfer $s$ to the receiver
whenever action $c$ is played, and makes no transfers
on other actions. Define
\[
    h(\theta)=u_S(\theta,c)-s,
    \qquad
    r(\theta,a)
    =u_R(\theta,a)+s\mathbf{1}_{\{a=c\}}.
\]
Thus, $r$ is the receiver's utility after accounting
for the sender's transfers, but before accounting
for the receiver's own commitments. Let
\[
    D=\max_{\theta\in\Theta,\ a,b\in A}
       |r(\theta,a)-r(\theta,b)|.
\]

Enumerate the actions as
$c=a_0,a_1,\ldots,a_{m-1}$.
Choose $G>6D$ sufficiently large that
\[
    u_S(\theta,a)-h(\theta)+G\ge0
    \qquad
    \text{for every $\theta$ and $a\ne c$},
\]
and put $K_{a_j}=jG$. In addition to the transfers
above, the sender commits to burn
\[
    u_S(\theta,a)-h(\theta)+K_a
\]
at every outcome $(\theta,a)$ with $a\ne c$, and
burns nothing when $c$ is played. All these amounts
are nonnegative. Before the receiver's commitment,
the sender's resulting utility is
\[
    h(\theta)-K_a.
\]
In particular, her ranking of actions is strict and
independent of the state.

If the receiver commits to no payments, every
continuation equilibrium plays $c$. Indeed, favorable
receiver tie-breaking makes the response to every
message pure: among receiver-optimal actions, the
sender strictly prefers the one with smallest $K_a$.
Since the sender ranks these responses identically
in every state, all on-path messages induce the same
action. That action must be optimal under the prior.
The unique prior-optimal action for the receiver is
$c$, since
\[
    \mathbb{E}_{\theta}[r(\theta,c)]
    =b_R(\Gamma)+\varepsilon,
    \qquad
    \mathbb{E}_{\theta}[r(\theta,a)]
    =R_a\le b_R(\Gamma)
    \quad(a\ne c).
\]
The receiver can therefore obtain
$b_R(\Gamma)+\varepsilon$ by making no payments.

We next show that any optimal receiver commitment
also results in action $c$, with no realized receiver
payments or burns. Fix an arbitrary receiver
commitment and a continuation equilibrium. Let $Q$
and $H$ denote, respectively, the receiver's expected
transfer to the sender and expected burning.

For each message $m$, let
\[
    k_m=\sum_{a\in A}\sigma_R(a\mid m)K_a,
    \qquad
    k_*=\min_{m\in M}k_m.
\]
The minimum includes off-path messages. By always
sending a message attaining $k_*$, the sender obtains
at least $\mathbb{E}_{\theta}[h(\theta)]-k_*$,
because incoming transfers are nonnegative. Her
equilibrium payoff is
$\mathbb{E}_{\theta}[h(\theta)]-\mathbb{E}[K_a]+Q$.
Consequently,
\begin{equation}\label{eq:sender-first-rank-incentive}
    \mathbb{E}[K_a]-k_*\le Q.
\end{equation}

For an on-path message $m$, let $L_m$ be the smallest
$K_a$ among actions played with positive probability
after $m$. Let $q_a(m)$ be the receiver's expected
transfer to the sender conditional on message $m$
and action $a$. Favorable receiver tie-breaking
implies that all actions in the receiver's support
give the same conditional sender payoff. Comparing
a supported action $a$ with a supported action
whose penalty is $L_m$ gives
\[
    q_a(m)\ge K_a-L_m.
\]
Averaging over messages and actions yields
\begin{equation}\label{eq:sender-first-rank-mixing}
    \mathbb{E}[k_m-L_m]\le Q.
\end{equation}

Since $k_m\ge k_*$ and every action played after $m$
has $K_a\ge L_m$, we have
\[
    (k_*-K_a)_+\le k_m-L_m
\]
on every realized outcome. Using
\eqref{eq:sender-first-rank-incentive} and
\eqref{eq:sender-first-rank-mixing}, we obtain
\[
\begin{aligned}
    \mathbb{E}|K_a-k_*|
    &=
    \mathbb{E}[K_a]-k_*
    +2\mathbb{E}[(k_*-K_a)_+]\\
    &\le 3Q.
\end{aligned}
\]
Choose an action $j$ whose $K_j$ is closest to $k_*$.
Because distinct penalties are separated by at
least $G$, every other action has distance at least
$G/2$ from $k_*$. Therefore,
\[
    \Pr(a\ne j)\le\frac{6Q}{G}.
\]

The receiver's expected utility is consequently
bounded by
\[
\begin{aligned}
    \mathbb{E}[r(\theta,a)]-Q-H
    &\le
    \mathbb{E}_{\theta}[r(\theta,j)]
    +D\Pr(a\ne j)-Q-H\\
    &\le
    \mathbb{E}_{\theta}[r(\theta,j)]
    -\left(1-\frac{6D}{G}\right)Q-H\\
    &\le
    b_R(\Gamma)+\varepsilon
    -\left(1-\frac{6D}{G}\right)Q-H.
\end{aligned}
\]
Since $G>6D$ and the receiver can obtain
$b_R(\Gamma)+\varepsilon$ by making no payments,
an optimal receiver commitment must have $Q=H=0$.
The preceding bounds then imply that a single action
is played almost surely, and that this action is $c$.

Thus, regardless of the optimal receiver response,
the sender obtains
\[
    S_c-s
    =W_{\mathrm{const}}(\Gamma)
     -b_R(\Gamma)-\varepsilon.
\]
This commitment deviation is available for every
$\varepsilon>0$. Hence every equilibrium of
$\Gamma^{S\to R}$ gives the sender at least
$W_{\mathrm{const}}(\Gamma)-b_R(\Gamma)$.

The receiver can guarantee at least $b_R(\Gamma)$
after every sender commitment by making no payments
and subsequently choosing a prior-best original
action, independently of the sender's message.
Incoming sender transfers are nonnegative.
Adding the two guarantees gives
\[
    \mathrm{EQ}_{S\to R,w}^{-}(\Gamma)
    \ge W_{\mathrm{const}}(\Gamma).
\]

To prove equality, prescribe a babbling continuation
after every completed commitment history, and let
the receiver make no payments after every sender
commitment. This is an optimal receiver response:
for any fixed sender commitment, a babbling
continuation following a receiver deviation gives
the receiver at most
\[
    \max_{a\in A}
    \mathbb{E}_{\theta}
    [u_R(\theta,a)+\tau_S^R(\theta,a)],
\]
which is precisely what she obtains by making no
payments.

The sender chooses $c\in\arg\max_a W_a$ and transfers
$b_R(\Gamma)-R_c$ on action $c$ in every state,
with no other payments or burns. Action $c$ is
receiver-optimal under the prior. Among any other
prior-optimal actions, favorable tie-breaking selects
one giving the same maximal sender payoff, or selects
$c$. The resulting utilities are
\[
    W_{\mathrm{const}}(\Gamma)-b_R(\Gamma)
    \quad\text{and}\quad
    b_R(\Gamma).
\]
No sender deviation is profitable: under the
prescribed babbling continuation, the receiver
obtains at least $b_R(\Gamma)$, while total welfare
is at most $W_{\mathrm{const}}(\Gamma)$.

This defines a sequential equilibrium attaining
$W_{\mathrm{const}}(\Gamma)$ welfare. Computing the
prior payoffs of all actions, selecting $c$, and
specifying its transfer requires $O(n|A|)$
arithmetic operations.

\section{Proof of Proposition~\ref{prop:bilateral-efficiency-obstruction}}\label{sec:prop-bilateral-efficiency-obstruction}

Consider the game $\Gamma$ with equally likely states
$\Theta=\{L,R\}$, actions $A=\{0,1\}$, and payoffs
written in the order $(u_S,u_R)$:
\[
\begin{array}{c|cc}
    &0&1\\ \hline
    L&(0,2)&(3,0)\\
    R&(3,0)&(0,2).
\end{array}
\]
Both constant actions give expected utilities
\[
    b_S(\Gamma)=\frac32,
    \qquad
    b_R(\Gamma)=1.
\]
Maximum welfare is attained by playing $1$ in state
$L$ and $0$ in state $R$, giving
\[
    W^*(\Gamma)=3.
\]

We first show that every sequential equilibrium under
either simultaneous or sender-first commitments has
welfare at most $5/2$.

Let $x$ be the equilibrium probability that the sender's
originally preferred action is played. The players'
original expected utilities are then $3x$ for the sender
and $2(1-x)$ for the receiver.
Let $T$ and $Q$ denote expected realized transfers from
sender to receiver and from receiver to sender,
respectively, and let $B_S,B_R\geq0$ denote their expected
realized burning. Their equilibrium utilities are
\[
    U_S=3x-T+Q-B_S,
    \qquad
    U_R=2(1-x)+T-Q-B_R.
\]

The sender can guarantee at least $3/2$ by committing
to no payments. To see this, consider any continuation
equilibrium after such a commitment, including any
subsequent receiver commitment in the sender-first
protocol. By sending a fixed message independently
of the state, the sender induces a state-independent
action distribution. Every such distribution gives
her original expected utility $3/2$, and incoming
receiver transfers are nonnegative.
Her equilibrium continuation payoff must therefore
be at least $3/2$. Hence,
\[
    3x-T+Q-B_S\geq\frac32,
\]
or equivalently,
\[
    T\leq3x+Q-B_S-\frac32.
    \tag{1}
\]

The receiver can deviate to no payments while holding
the sender's commitment fixed under either protocol.
For each constant action $a$, define
\[
    t_a
    :=
    \mathbb{E}_{\theta\sim p}
    [\tau_S^R(\theta,a)].
\]
Each constant action gives the receiver original
expected utility $1$. Moreover, nonnegativity of
transfers implies
\[
    t_0+t_1\geq T,
\]
because the left-hand side includes every sender
transfer realized under the equilibrium allocation.
After committing to no payments, the receiver can
therefore guarantee
\[
    1+\max\{t_0,t_1\}\geq1+\frac T2.
\]
Consequently,
\[
    2(1-x)+T-Q-B_R\geq1+\frac T2,
\]
which gives
\[
    T\geq4x+2Q+2B_R-2.
    \tag{2}
\]

Combining (1) and (2) yields
\[
    x+Q+B_S+2B_R\leq\frac12.
\]
In particular, $x\leq1/2$, and equilibrium welfare
satisfies
\[
    U_S+U_R
    =2+x-B_S-B_R
    \leq\frac52.
\]

It remains to show that this bound is attained under
both protocols. Throughout the following constructions,
prescribe a babbling continuation equilibrium after
every commitment profile, with the receiver choosing
a pure best constant action.

For simultaneous commitments, let both players commit
to no payments. A unilateral deviation by the sender
gives at most $3/2$: the receiver still makes no
transfers, every constant action gives the sender
original expected utility $3/2$, and her own payments
only reduce that utility. Similarly, a receiver
deviation gives at most $1$. Thus, zero commitments
followed by babbling form a sequential equilibrium
with welfare $5/2$.

For sender-first commitments, prescribe that the
receiver commits to no payments after every sender
commitment. This is an optimal response under the
specified babbling continuations. Indeed, fixing
the sender's commitment, the receiver can obtain
\[
    \max_{a\in\{0,1\}}(1+t_a)
\]
by making no payments and choosing a best constant
action. No receiver commitment can yield more in
a babbling continuation, since its own payments
only subtract from these constant-action payoffs.
Anticipating this response, the sender obtains at
most $3/2$ from any commitment and attains $3/2$
by committing to no payments. This again defines
a sequential equilibrium with welfare $5/2$.

Therefore,
\[
    \mathrm{EQ}_{SR,w}^{+}(\Gamma)
    =
    \mathrm{EQ}_{S\to R,w}^{+}(\Gamma)
    =
    b_S(\Gamma)+b_R(\Gamma)
    =
    \frac52
    <3
    =
    W^*(\Gamma).
\]

\section{Receiver-First Commitments: Shared Constructions}
\label{sec:receiver-first-constructions}

We first establish two implementation lemmas that apply
to arbitrary finite action sets. These will be used both
to prove the general efficiency theorem and to characterize
equilibrium payoffs with binary actions.

For a receiver commitment $\rho_R$, write
\[
    S_x^\rho
    :=\mathbb{E}_{\theta}
      [u_S^{(\mathbf{0},\rho_R)}(\theta,x)],
    \qquad
    R_x^\rho
    :=\mathbb{E}_{\theta}
      [u_R^{(\mathbf{0},\rho_R)}(\theta,x)],
\]
\[
    W_x^\rho:=S_x^\rho+R_x^\rho,
    \qquad
    b^\rho:=\max_{x\in A}R_x^\rho.
\]
Thus these benchmarks incorporate the receiver's commitment
but no sender payments.

\begin{lemma}\label{lem:receiver-first-babbling}
Fix any receiver commitment $\rho_R$. If babbling is
prescribed after every subsequent sender commitment,
the sender has an optimal response giving utilities
\[
    \bigl(V(\rho_R),b^\rho\bigr),
    \qquad
    V(\rho_R):=\max_{x\in A}W_x^\rho-b^\rho.
\]
Moreover, $b^\rho\leq b_R(\Gamma)$.
\end{lemma}

\begin{proof}
In any babbling continuation, the receiver obtains at
least $b^\rho$, since incoming sender transfers are
nonnegative. Constant-action welfare, including any
sender burning, is at most $\max_x W_x^\rho$.
Thus the sender's payoff is at most $V(\rho_R)$.

To attain this bound, choose
$c\in\arg\max_x W_x^\rho$ and let the sender transfer
$b^\rho-R_c^\rho$ on action $c$ in every state, with
no other payments. Action $c$ is then receiver-optimal
under the prior. Any other receiver-optimal action
$x\neq c$ satisfies $R_x^\rho=b^\rho$ and
\[
    S_x^\rho=W_x^\rho-b^\rho
    \leq W_c^\rho-b^\rho.
\]
Thus favorable tie-breaking permits a babbling continuation
playing $c$, which gives the claimed utilities.

Finally, the receiver's own payments only subtract from
its original utility, so
\[
    b^\rho
    =
    \max_{x\in A}
    \mathbb{E}_{\theta}
    \bigl[u_R(\theta,x)-\|\rho_R(\theta,x)\|\bigr]
    \leq b_R(\Gamma).
\]
\end{proof}

Lemma~\ref{lem:receiver-first-babbling} provides credible
continuation strategies after receiver deviations:
prescribe babbling after every sender commitment and
select the sender's optimal response described in the
lemma. The receiver then obtains at most its original
babbling payoff.

The next lemma implements nonconstant action rules.
Nonconstancy is useful because every action is off path
in some state. The receiver can therefore reduce the
attractiveness of alternative constant actions without
changing the prescribed payoffs.

\begin{lemma}\label{lem:nonconstant-implementation}
Let $\Gamma$ be a finite sender-receiver game, and let
$f:\Theta\to A$ be a nonconstant action rule. Write
\[
    W(f):=\sum_{\theta}p(\theta)w(\theta,f(\theta)),
    \qquad
    \underline S:=\min_{a\in A}S_a(\Gamma).
\]
Every utility profile $(U_S,U_R)$ satisfying
\[
    U_S\geq\underline S,
    \qquad
    U_R\geq b_R(\Gamma),
    \qquad
    U_S+U_R\leq W(f)
\]
is implemented by a sequential equilibrium of
$\Gamma^{R\to S}$ inducing $f$.

For rational input data and a rational target profile,
the equilibrium can be constructed in $O(n|A|)$
arithmetic operations and has a description of
polynomial encoding length.
\end{lemma}

\begin{proof}
Choose $a\in\arg\min_{x\in A}S_x(\Gamma)$, and let
\[
    E:=\{\theta:f(\theta)\neq a\}.
\]
Since $f$ is nonconstant, $E$ has positive probability.
Choose nonnegative numbers $k_\theta,d_\theta$, equal
to zero outside $E$, such that
\[
    \sum_\theta p(\theta)k_\theta=U_S-\underline S,
    \qquad
    \sum_\theta p(\theta)d_\theta=W(f)-U_S-U_R.
\]
For example, both sums can be concentrated on one
state in $E$.

Define the desired statewise payoffs
\[
    z_\theta:=u_S(\theta,a)+k_\theta,
    \qquad
    r_\theta:=w(\theta,f(\theta))-z_\theta-d_\theta.
\]
Their expectations are $U_S$ and $U_R$, respectively.

\paragraph{Payments on prescribed outcomes.}
At $(\theta,f(\theta))$, let the sender transfer
\[
    \bigl(u_S(\theta,f(\theta))-z_\theta\bigr)_+
\]
to the receiver, and let the receiver transfer
\[
    \bigl(z_\theta-u_S(\theta,f(\theta))\bigr)_+
\]
to the sender, where $y_+:=\max\{y,0\}$.
The receiver additionally burns $d_\theta$.
There is no sender burning on the prescribed outcome.
The resulting payoffs are exactly
$(z_\theta,r_\theta)$.

\paragraph{Payments on off-path outcomes.}
For each $x\neq f(\theta)$, make no transfers and let
the sender and receiver burn, respectively,
\[
    1+\max\{0,u_S(\theta,x)-z_\theta\},
    \qquad
    1+\max\{0,u_R(\theta,x)-r_\theta\}.
\]
Both players then strictly prefer $f(\theta)$ to
every other action at state $\theta$.

Let $\rho_R$ denote the receiver scheme constructed
so far. The receiver makes no transfers to the sender
on action $a$. Indeed, when $f(\theta)=a$, we have
$k_\theta=d_\theta=0$ and $z_\theta=u_S(\theta,a)$,
so neither player pays on the prescribed outcome.
When $f(\theta)\neq a$, action $a$ is off path and
carries no transfers. Consequently,
\[
    S_a^\rho=S_a(\Gamma)=\underline S.
\]

For each action $x\neq a$, choose a state $\theta_x$
with $f(\theta_x)\neq x$. Such a state exists because
$f$ is nonconstant. Increase the receiver's burn at
$(\theta_x,x)$ until
\[
    R_x^\rho<R_a^\rho
    \qquad\text{and}\qquad
    W_x^\rho<W_a^\rho.
\]
These additional burns are off path, preserve the
statewise incentive comparisons, and do not affect
either benchmark for action $a$. They imply
\[
    V(\rho_R)
    =W_a^\rho-R_a^\rho
    =\underline S.
\]

\paragraph{Equilibrium strategies.}
At the constructed commitment profile, the sender
recommends $f(\theta)$ and the receiver follows.
Both players strictly prefer the prescribed action
in every state. Thus following each on-path recommendation
is strictly optimal for the receiver, and no sender
type benefits from another message. At unused messages,
assign prior beliefs and a receiver best response.
These beliefs are consistent with state-independent
trembles. This defines a continuation equilibrium
with expected utilities $(U_S,U_R)$.

After every different sender commitment following
$\rho_R$, prescribe babbling. By
Lemma~\ref{lem:receiver-first-babbling}, any such deviation
gives the sender at most
$V(\rho_R)=\underline S\leq U_S$.
The prescribed sender commitment is therefore optimal.

After every different receiver commitment, prescribe
babbling after every sender commitment and select
the optimal sender response from
Lemma~\ref{lem:receiver-first-babbling}.
These continuations give the receiver at most
$b_R(\Gamma)\leq U_R$. Its prescribed commitment is
therefore optimal as well.

\paragraph{Computation.}
The benchmarks and prescribed payments require
$O(n|A|)$ arithmetic operations. Using the benchmarks
before the additional burns, the increase at
$(\theta_x,x)$ can be chosen as
\[
    \frac{
        1+\max\{0,R_x^\rho-R_a^\rho,
                   W_x^\rho-W_a^\rho\}
    }{p(\theta_x)}.
\]
Finding the states $\theta_x$ also takes at most
$O(n|A|)$ operations. Off-path behavior is represented
by the rules above. All amounts are obtained through
sums, differences, maxima, and divisions by positive
input probabilities, so their encoding lengths are
polynomial in the rational input length.
\end{proof}

\section{Proof of Theorem~\ref{thm:opt-welfare}}\label{sec:thm-opt-welfare}

Transfers preserve welfare and burning reduces it,
so $W^*(\Gamma)$ is an upper bound. We construct an
equilibrium attaining it.

Suppose first that some constant action $c$ is efficient,
so $W_c(\Gamma)=W^*(\Gamma)$. Let the receiver commit
to no payments. After every receiver commitment,
prescribe babbling after all sender commitments and
select an optimal sender response as in
Lemma~\ref{lem:receiver-first-babbling}.

At the zero receiver scheme, choose the optimal sender
response inducing $c$. It gives utilities
\[
    \bigl(W^*(\Gamma)-b_R(\Gamma),\,b_R(\Gamma)\bigr).
\]
The sender's response is optimal, and any receiver
deviation gives at most $b_R(\Gamma)$ by the lemma.
Thus this is an efficient equilibrium.

Otherwise, choose a statewise efficient rule $f$,
which must be nonconstant, and let
\[
    \underline S:=\min_{a\in A}S_a(\Gamma).
\]
If $j$ is an original receiver-best constant action,
then
\[
    W^*(\Gamma)\geq W_j(\Gamma)
    =S_j(\Gamma)+b_R(\Gamma)
    \geq\underline S+b_R(\Gamma).
\]
Hence the profile
\[
    \bigl(\underline S,\,
          W^*(\Gamma)-\underline S\bigr)
\]
satisfies the conditions of
Lemma~\ref{lem:nonconstant-implementation}.
Applying that lemma to $f$ yields an efficient
equilibrium.

An efficient rule, all constant-action benchmarks,
and the resulting equilibrium description can be
computed in $O(n|A|)$ arithmetic operations.
This proves the theorem.

\section{Proof of Theorem~\ref{thm:receiver-first-payoffs}}\label{sec:thm-receiver-first-payoffs}

Throughout this proof, $A=\{0,1\}$.
We use the notation and constructions of
Section~\ref{sec:receiver-first-constructions}.

In addition to the babbling-response calculation,
Proposition~\ref{prop:sen-pessimistic} provides a guarantee against arbitrary continuation selection. Applying
Proposition~\ref{prop:sen-pessimistic} to the game
after a fixed receiver commitment $\rho_R$ shows that
the sender can guarantee
\[
    V(\rho_R)=\max_{a\in\{0,1\}}W_a^\rho-b^\rho.
\]
Therefore, her payoff following $\rho_R$ in any
equilibrium of the commitment extension is at least
$V(\rho_R)$.

\paragraph{Necessary conditions.}
Fix a sequential equilibrium with expected utilities
$(U_S,U_R)$.

The receiver can commit to no payments and then
guarantee its original best constant-action payoff,
since incoming sender transfers are nonnegative.
Thus
\[
    U_R\geq\overline R.
\]

After any receiver commitment, the sender can commit
to no payments. In any resulting continuation equilibrium,
sending a fixed message independently of the state
induces a state-independent action distribution.
Her original expected utility from this distribution
is at least $\underline S$, and incoming transfers
are nonnegative. Consequently,
\[
    U_S\geq\underline S.
\]

We next show that the equilibrium continuation can
be replaced by one inducing a deterministic action
rule without changing either player's expected payoff.
All preferences in this argument refer to utilities
after the fixed equilibrium commitments.

Let $\ell$ and $h$ be the smallest and largest
probabilities of action $0$ induced by messages.
Suppose first that $\ell<h$. Sender optimality,
including favorable tie-breaking, implies that every
state in which at least one player has a strict
preference uses messages inducing the appropriate
extreme probability.

Replace the receiver's response at every lower-extreme
message by action $1$, and at every upper-extreme
message by action $0$. This preserves sender optimality.
Whenever a response changes, the receiver was mixing
and hence indifferent between actions. Favorable
tie-breaking also implies sender indifference in
conditional expected utility at that message.
Thus these changes preserve both players' expected
payoffs.

States in which both players are indifferent may use
intermediate messages or randomize between extremes.
Assign all these types to one upper-extreme message
and update on-path beliefs by Bayes' rule. Their
contribution to either player's unnormalized payoff
difference between actions is zero, so receiver
optimality and favorable tie-breaking are preserved.
The reassigned types remain indifferent.
Messages that become unused can retain their previous
beliefs and their responses after purification;
appropriate trembles support these beliefs.
The resulting rule is deterministic and preserves
both expected payoffs.

If $\ell=h$, the action distribution is state-independent.
If it is mixed, receiver optimality and favorable
tie-breaking imply that both players are indifferent
under the prior between the two constant actions.
A pure babbling continuation therefore gives the same
payoffs. Thus a deterministic rule can be selected
in this case as well.

Because the replacement preserves payoffs at the
equilibrium commitment profile and leaves all other
continuations unchanged, commitment-stage incentives
are preserved.

If the resulting deterministic rule is nonconstant,
its original welfare is at most $W_{\mathrm{nc}}(\Gamma)$.
Transfers preserve welfare and burning reduces it,
so
\[
    U_S+U_R\leq W_{\mathrm{nc}}(\Gamma).
\]
Together with the individual lower bounds, this gives
\[
    (U_S,U_R)\in\mathcal U_{\mathrm{nc}}(\Gamma).
\]

Suppose instead that the rule always plays action $c$.
Let $\rho_R$ be the equilibrium receiver commitment,
and let $B_i(c)$ denote player $i$'s expected burning
on constant action $c$. Then
\[
    U_S+U_R=W_c-B_R(c)-B_S(c),
    \qquad
    W_c^\rho=W_c-B_R(c).
\]
The sender's guarantee gives
\[
    U_S\geq V(\rho_R)
    \geq W_c-B_R(c)-b^\rho.
\]
Combining these inequalities yields
\[
    U_R+B_S(c)\leq b^\rho\leq\overline R.
\]
Since $U_R\geq\overline R$, we obtain
\[
    U_R=\overline R,
    \qquad
    B_S(c)=0,
    \qquad
    b^\rho=\overline R.
\]

Choose an action $j$ attaining $b^\rho$. Nonnegative
receiver payments imply
\[
    R_j^\rho\leq R_j\leq\overline R.
\]
Equality therefore holds throughout, and the receiver
makes no payments on action $j$ in any state.
In particular, $W_j^\rho=W_j$, so
\[
    U_S\geq V(\rho_R)
    \geq W_j^\rho-b^\rho
    =S_j\geq\underline S_R.
\]
Moreover,
\[
    U_S
    =W_c-B_R(c)-\overline R
    \leq W_{\mathrm{const}}(\Gamma)-\overline R.
\]
Thus
\[
    (U_S,U_R)\in\mathcal U_{\mathrm c}(\Gamma),
\]
proving necessity.

\paragraph{Implementing the nonconstant region.}
Fix $(U_S,U_R)\in\mathcal U_{\mathrm{nc}}(\Gamma)$.
Choose a nonconstant rule $f$ attaining
$W_{\mathrm{nc}}(\Gamma)$. The defining inequalities
of this region are exactly the conditions of
Lemma~\ref{lem:nonconstant-implementation}.
That lemma therefore implements $(U_S,U_R)$.

\paragraph{Implementing the constant region.}
Fix $(U_S,\overline R)\in\mathcal U_{\mathrm c}(\Gamma)$.
Choose an original receiver-best action $j$ with
$S_j=\underline S_R$, and choose
$c\in\arg\max_a W_a$. Put
\[
    D:=W_{\mathrm{const}}(\Gamma)-\overline R-U_S
    \geq0.
\]

If $c=j$, the inequalities defining the constant
region imply $U_S=S_j$ and $D=0$; let the receiver
commit to no payments. Otherwise, let the receiver
burn $D$ on action $c$ in every state and make no
other payments. Denote this scheme by $\rho_R$.

Action $j$ retains receiver payoff $\overline R$,
and no action has a higher receiver payoff, so
\[
    b^\rho=\overline R.
\]
Also,
\[
    W_c^\rho
    =W_{\mathrm{const}}(\Gamma)-D
    =U_S+\overline R.
\]
If $c\neq j$, binary actions imply that $j$ is the
only other action, and
\[
    W_j^\rho=S_j+\overline R
    \leq U_S+\overline R=W_c^\rho.
\]
If $c=j$, it already maximizes constant-action welfare.
Hence in either case,
\[
    V(\rho_R)=U_S.
\]

Prescribe babbling after every sender commitment and
choose the optimal sender response inducing $c$
from Lemma~\ref{lem:receiver-first-babbling}.
The resulting utilities are $(U_S,\overline R)$.
After every other receiver commitment, prescribe
the optimal babbling response from the same lemma.
No receiver deviation then yields more than
$\overline R$.

This implements every profile in
$\mathcal U_{\mathrm c}(\Gamma)$ and completes
the characterization.

\section{Proof of Proposition~\ref{prop:persuasion-equivalence}}\label{sec:proof-persuasion-equivalence}

We first show that every equilibrium under a privately
randomized sender commitment is bounded by an appropriate
persuasion scheme. We then implement an optimal persuasion
scheme using a private lottery over payment schemes.

\paragraph{From commitments to persuasion.}
Fix any privately randomized sender commitment and
continuation equilibrium. Let
\[
    x_{\theta a}
    :=\Pr(\text{receiver chooses }a\mid\theta),
    \qquad
    \alpha_a:=\sum_\theta p(\theta)x_{\theta a},
\]
and let $P_a$ be the unconditional expected transfer
actually paid to the receiver on outcomes where
action $a$ is chosen.

For every pair of actions $a,b$, aggregate the
receiver's optimality inequalities over messages
at which action $a$ is chosen, weighting by the
probability of that message and action. Since
transfers promised on deviations are nonnegative,
the resulting inequality implies
\[
    \sum_\theta p(\theta)x_{\theta a}
    \bigl(u_R(\theta,a)-u_R(\theta,b)\bigr)
    +P_a
    \geq0.
\]
Indeed, omitting the receiver's incoming transfer
on the alternative action can only increase the
left-hand side.

If $\alpha_a>0$, set $z_a:=P_a/\alpha_a$.
If $\alpha_a=0$, then $P_a=0$, and set $z_a:=0$.
The direct persuasion scheme recommending $a$
with probability $x_{\theta a}$ and paying $z_a$
for compliance is therefore obedient. Its sender
payoff is
\[
    \sum_{\theta,a}p(\theta)x_{\theta a}u_S(\theta,a)
    -\sum_a P_a.
\]
This is at least the sender's equilibrium payoff
under the original commitment, since any realized
burning further reduces that payoff. Consequently,
\[
    \OPT_{S,\Delta}^{+}(\Gamma)
    \leq\mathrm{BP}_{+}(\Gamma).
\]

If the commitment permits only burning, all $P_a$
are zero. The same argument then gives an obedient
persuasion scheme without payments, proving
\[
    \OPT_{S,\Delta,\bot}^{+}(\Gamma)
    \leq\mathrm{BP}(\Gamma).
\]

\paragraph{From persuasion to commitments.}
Consider an obedient direct persuasion scheme
$\pi(a\mid\theta)$ with nonnegative compliance
payments $z_a$. Represent its recommendation rule
as a lottery over deterministic maps $f:\Theta\to A$.
For example, independently select $f(\theta)$
according to $\pi(\cdot\mid\theta)$ for every state.
The resulting distribution over maps is finite,
independent of the realized state, and satisfies
\[
    \Pr(f(\theta)=a)=\pi(a\mid\theta).
\]

For each map $f$, define a sender payment scheme
$\tau^f$ as follows. At state $\theta$, let
$a=f(\theta)$. On action $a$, the sender transfers
$z_a$ to the receiver and burns nothing.
On every other action $b\neq a$, she makes no
transfer and burns
\[
    1+\max\{0,u_S(\theta,b)-u_S(\theta,a)+z_a\}.
\]
Thus action $a$ gives the sender
$u_S(\theta,a)-z_a$, strictly more than any other
action under the realized scheme.

The sender publicly commits to the lottery over
these schemes and privately observes its realization.
She then recommends the prescribed action $f(\theta)$.
Following recommendation $a$, the receiver obtains
the same conditional payoffs as in the persuasion
scheme: it receives $z_a$ if it follows and no
transfer otherwise. Obedience therefore follows
from the persuasion constraints.

If the receiver is indifferent, favorable tie-breaking
also selects the recommendation, since every sender
type compatible with that recommendation strictly
prefers compliance under its realized scheme.
The sender has no profitable message deviation.
At unused messages, prior beliefs and receiver
best responses complete the assessment; these
beliefs are consistent with state-independent
trembles. Hence the constructed strategies form
a sequential equilibrium of the continuation game.

No burning occurs on path, so the sender receives
exactly the persuasion payoff. Setting all $z_a=0$
gives the same implementation using burning alone.

Optimal persuasion schemes exist here. In the
payment model, one can use the variables
\[
    P_a=\alpha_a z_a
\]
to express the problem as a finite linear program:
\[
\begin{aligned}
\max_{x,P}\quad&
    \sum_{\theta,a}p(\theta)x_{\theta a}u_S(\theta,a)
    -\sum_a P_a\\
\text{subject to}\quad&
    \sum_a x_{\theta a}=1
    &&\text{for every }\theta,\\
&
    \sum_\theta p(\theta)x_{\theta a}
    \bigl(u_R(\theta,a)-u_R(\theta,b)\bigr)+P_a
    \geq0
    &&\text{for every }a\neq b,\\
&
    x_{\theta a}\geq0,\qquad P_a\geq0.
\end{aligned}
\]
It is feasible and has a finite optimal value.
At an optimum, $\alpha_a=0$ implies $P_a=0$,
so the compliance payments can be recovered as
above. Setting $P_a=0$ for every action gives
the ordinary persuasion problem.

Finally, implement an optimal persuasion scheme
and prescribe a babbling continuation after every
other announced commitment. The upper bounds
already proved apply to every such deviation.
Thus the constructed commitment is itself optimal,
and the implementation extends to an equilibrium
of the entire commitment game. This proves both
equalities and attainment.

\section{Proof of Corollary~\ref{cor:general-randomized-computation}}
\label{app:general-randomized-computation}
The linear program in the proof of
Proposition~\ref{prop:persuasion-equivalence} has
$nm+m$ variables and $O(n+m^2)$ constraints,
apart from nonnegativity constraints. It therefore
admits polynomial-time optimization for rational
input data. Setting all payment variables to zero
gives the corresponding linear program for
burning-only commitments.

Let $\pi(a\mid\theta)$ be an optimal recommendation
rule, and let $z_a$ be its nonnegative compliance
payments, with $z_a=0$ in the burning-only case.
We construct a small lottery over deterministic
action rules having the same recommendation
probabilities.

Fix an ordering $a_1,\ldots,a_m$ of the actions.
For each state $\theta$, define the cumulative
probabilities
\[
    c_{\theta,0}=0,
    \qquad
    c_{\theta,k}
    =\sum_{\ell=1}^{k}\pi(a_\ell\mid\theta)
    \quad(k=1,\ldots,m).
\]
Collect all these numbers, remove repetitions,
and sort them:
\[
    0=t_0<t_1<\cdots<t_K=1.
\]
There are at most $n(m-1)$ distinct interior
breakpoints, so $K\le n(m-1)+1$.

Each interval $(t_{j-1},t_j)$ defines a deterministic
action rule $f_j$: in state $\theta$, select the
action whose cumulative-probability interval
contains $(t_{j-1},t_j)$.
Choose $f_j$ with probability $t_j-t_{j-1}$.
By construction,
\[
    \Pr(f_j(\theta)=a)=\pi(a\mid\theta)
    \qquad
    \text{for every $\theta$ and $a$}.
\]

Apply the implementation construction of
Proposition~\ref{prop:persuasion-equivalence} to
each $f_j$, using the compliance payments $z_a$.
The resulting lottery has at most $K$ payment
schemes and attains the optimal persuasion payoff.
By that proposition, prescribing babbling
continuations after other announced commitments
extends the construction to a sequential equilibrium
of the commitment game.

An optimal rational solution of the linear program
has polynomial encoding length. Computing the
compliance payments, sorting the breakpoints,
and constructing the payment schemes all require
polynomial time and produce rational parameters
of polynomial encoding length.

\end{document}